\documentclass[conference]{IEEEtran}

\usepackage{graphicx}
\usepackage{xspace}
\let\labelindent\relax
\usepackage{enumitem}
\usepackage{balance}

\ifCLASSINFOpdf
\else
\fi

\usepackage[cmex10]{amsmath}
\usepackage{amssymb,amsthm}
\usepackage[noadjust]{cite}
\usepackage[tight,footnotesize]{subfigure}
\usepackage{booktabs}

\newtheorem{theorem}{Theorem}

\newtheorem{lemma}{Lemma} 

\newtheorem{proposition}{Proposition} 
\newtheorem{example}{Example}
\newtheorem{remark}{Remark}

\newcommand{\notimplies}{\mathrel{{\ooalign{\hidewidth$\not\phantom{=}$\hidewidth\cr$\implies$}}}}

\newcommand*{\defeq}{\mathrel{\vcenter{\baselineskip0.5ex \lineskiplimit0pt
                     \hbox{\scriptsize.}\hbox{\scriptsize.}}}
                     =}

\newcommand*{\GP}{\ensuremath{\mathbf{(GP)}}}
\newcommand*{\WS}{\ensuremath{\mathbf{(S)}}}
\newcommand*{\SR}{\ensuremath{\mathbf{(I)}}}
\newcommand*{\Mo}{\ensuremath{\mathbf{(M)}}}
\newcommand*{\SMo}{\ensuremath{\mathbf{(SM)}}}
\newcommand*{\LN}{\ensuremath{\mathbf{(LP)}}}
\newcommand*{\Idty}{\ensuremath{\mathbf{(Id)}}}
\newcommand*{\TM}{\ensuremath{\mathbf{(TM)}}}
\newcommand*{\PM}{\ensuremath{\mathbf{(PM)}}}

\newcommand*{\PMU}{\ensuremath{\mathbf{(PM_{u})}}}
\newcommand*{\PMUC}{\ensuremath{\mathbf{(PM_{u}^c)}}}
\newcommand*{\TMU}{\ensuremath{\mathbf{(TM_{u})}}}

\newcommand*{\op}[1]{ \ensuremath{ \operatorname{#1} }}
\newcommand*{\hgr}[1]{ \ensuremath{ \operatorname{hgr}\left( #1 \right)}}

\begin{document}

\setlength{\abovedisplayskip}{.2cm}
\setlength{\belowdisplayskip}{.2cm}

\title{Synergy, Redundancy and Common Information}

\author{
\IEEEauthorblockN{Pradeep Kr. Banerjee}
\IEEEauthorblockA{Dept. of E \& ECE\\ Indian Institute of Technology Kharagpur\\
pradeep.banerjee@gmail.com}
\and
\IEEEauthorblockN{Virgil Griffith}
\IEEEauthorblockA{School of Computing\\ National University of Singapore\\
dcsvg@nus.edu.sg}}

\maketitle

\begin{abstract}
We consider the problem of decomposing the total mutual information conveyed by a pair of predictor random variables about a target random variable into redundant, unique and synergistic contributions. We focus on the relationship between ``redundant information'' and the more familiar information-theoretic notions of ``common information.'' Our main contribution is an impossibility result.  We show that for independent predictor random variables, any common information based measure of redundancy cannot induce a nonnegative decomposition of the total mutual information. Interestingly, this entails that any reasonable measure of redundant information cannot be derived by optimization over a single random variable.
\end{abstract}


\IEEEpeerreviewmaketitle

\section{Introduction}
\vspace{2mm}
A complex system consists of multiple interacting parts or subsystems. A prominent example is the human brain that exhibits structure spanning a hierarchy of multiple spatial and temporal scales \cite{ref1}. A series of recent papers have focused on the problem of information decomposition in complex systems \cite{ref2,ref3,ref4,ref5,ref6,ref7,ref8,ref9,ref10,ref23}. A simple version of the problem can be stated as follows: The total mutual information that a pair of \emph{predictor} random variables (RVs) $({{X}_{1}},{{X}_{2}})$ convey about a \emph{target} RV $Y$ can have aspects of \emph{synergistic information} (conveyed only by the joint RV $({{X}_{1}}{{X}_{2}})$), of \emph{redundant} information (\emph{identically} conveyed by both ${{X}_{1}}$ and ${{X}_{2}}$), and of unique or private information (\emph{exclusively} conveyed by either ${{X}_{1}}$ or ${{X}_{2}}$). Is there a principled information-theoretic way of decomposing the total mutual information $I(X_1X_2;Y)$ into nonnegative quantities? 

Developing a principled approach to disentangling synergy and redundancy has been a long standing pursuit in neuroscience and allied fields\footnote{We invite the interested reader to see Appendix B, where we provide a sampling of several interesting examples and applications, where information-theoretic notions of synergy and redundancy are deemed useful.} \cite{ref1}, \cite{ref15,ref24,ref28,ref29,ref30,ref90}. However, the traditional apparatus of Shannon's information theory does not furnish ready-made tools for quantifying multivariate interactions. Starting with the work of Williams and Beer \cite{ref2}, several workers have begun addressing these issues \cite{ref3,ref4,ref5,ref6,ref7,ref8,ref9,ref10,ref23}. For the general case of $K$ predictors, Williams and Beer \cite{ref2} proposed the partial information (PI) decomposition framework to specify how the total mutual information about the target is shared across the singleton predictors and their overlapping or disjoint coalitions. Effecting a nonnegative decomposition has however turned out to be a surprisingly difficult problem even for the modest case of $K = 3$ \cite{ref3}, \cite{ref6}. Furthermore, there seems to be no clear consensus as to what is an ideal measure of redundant information. 

We focus on the relationship between redundant information and the more familiar information-theoretic notions of common information \cite{ref31,ref32}.
We distinguish synergistic and redundant interactions that exist within a group of predictor RVs from those that exist between a group of predictor RVs and a target RV.
A popular measure of the former (symmetric) type of interaction is the co-information \cite{ref46}.
Our main interest, however, lies in asymmetric measures of interaction that distinguish the target RV from the group of predictor RVs. 
An instance of such an interaction is when populations of retinal ganglion cells (predictors) interact to encode a (target) visual stimulus. \cite{ref58}. Yet another instance is when multiple genes (predictors) cooperatively interact within cellular pathways to specify a (target) phenotype \cite{ref24}.
In building up to our main contribution, we review and extend existing (symmetric) measures of common information to capture the asymmetric nature of these interactions. 

\emph{Section organization and summary of results}. In Section II, building on the heuristic notion of embodying information using $\sigma$-algebras and sample space partitions, we formalize the notions of common and private information structures. Information in the technical sense of entropy hardly captures the structure of information embodied in a source. First introduced by Shannon in a lesser known, short note \cite{ref35}, information structures capture the quintessence of ``information itself.'' We bridge several inter-related domains---notably, game theory, distributed control, and team decision problems to investigate the properties  of such structures. Surprisingly, while the ideas are not new, we are not aware of any prior work or exposition where common and private information structures have received a unified treatment. For instance, the notion of common information structures have appeared independently in at least four different early works, namely, that of Shannon \cite{ref35}, G{\'a}cs and K\"{o}rner \cite{ref31}, Aumann \cite{ref27}, and Hexner and Ho \cite{ref38}, and more recently in \cite{ref25}, \cite{ref36}, \cite{ref48}. In the first part of (mostly expository) Section II, we make some of these connections explicit for a finite alphabet. 

In the second part of Section II, we take a closer look at the intricate relationships between a pair of RVs. Inspired by the notion of private information structures \cite{ref38}, we derive a measure of private information and show how a dual of that measure recovers a known result \cite{ref25} in the form of the minimal sufficient statistic for one variable with respect to the other. We also introduce two new measures of common information. The richness of the decomposition problem is already manifest in simple examples when common and private informational parts cannot be isolated.

In Section III, we inquire if a nonnegative PI decomposition of $I(X_1X_2;Y)$ can be achieved using a measure of redundancy based on the notions of common information due to G{\'a}cs and K\"{o}rner \cite{ref31} and Wyner \cite{ref32}. We answer this question in the negative. 
For independent predictor RVs when any nonvanishing redundancy can be attributed solely to a mechanistic dependence between the target and the predictors, we show that any common information based measure of redundancy cannot induce a nonnegative PI decomposition.

\section{Information Decomposition into Common and Private Parts: The Case for Two Variables}
\vspace{2mm}

Let $(\Omega, \mathfrak{F}, \mathbb{P})$ be a fixed probability triple, where $\Omega$ is the set of all possible outcomes, elements of the $\sigma$-algebra $\mathfrak{F}$ are events and $\mathbb{P}$ is a function returning an event's probability. 
A random variable (RV) $X$ taking values in a discrete measurable space $(\mathcal{X},\mathfrak{X})$ (called the alphabet) is a measurable function $X:\Omega \to \mathcal{X}$ such that if $x \in \mathfrak{X}$, then $X^{-1}(x)=\{\omega:X(\omega) \in x\} \in \mathfrak{F}$. 
The $\sigma$-algebra induced by $X$ is denoted by $\sigma(X)$. We use ``iff'' as a shorthand for ``if and only if''.

\subsection{Information Structure Aspects}

The heuristic notion of embodying information using $\sigma$-algebras is not new \cite{ref50}\footnote{See though Example 4.10 in \cite{ref50} for a counterexample.}, \cite{ref49}. A sense in which $\sigma(X)$ represents information is given by the following lemma (see Lemma 1.13 in \cite{ref47}).  

\begin{lemma}[Doob-Dynkin Lemma]
Let $X_1:\Omega \to \mathcal{X}_1$ and $X_2:\Omega \to \mathcal{X}_2$ be two RVs, where $(\mathcal{X}_2,\mathfrak{X}_2)$ is a standard Borel space. Then $X_2$ is $\sigma(X_1)$-measurable, or equivalently $\sigma(X_2) \subset \sigma(X_1)$ iff there exists a measurable mapping $f:\mathcal{X}_1\to\mathcal{X}_2$ such that $X_2=f(X_1)$.
\end{lemma}
Suppose an agent does not know the ``true'' point $\omega \in \Omega$ but only observes an outcome $X_1(\omega)$. If for each drawn $\omega$, he takes some decision $X_2(\omega)$, then clearly $X_1(\omega)$ determines $X_2(\omega)$ so that we necessarily have $X_2=f(X_1)$. The Doob-Dynkin lemma says that this is equivalent to $X_2$ being $\sigma(X_1)$-measurable under some reasonable assumptions on the underlying measurable spaces.

From Lemma 1, it is easy to see that $X_1$ and $X_2$ carry the ``same information'' iff $\sigma(X_1)=\sigma(X_2)$. This notion of \emph{informational sameness} (denoted $X_1 \sim^{S} X_2$) induces a partition on the set of all RVs into equivalence classes called \emph{information elements}. We say that the RV $X_S$ is \emph{representative} of the information element $S$. First introduced by Shannon in a (perhaps) lesser known, short note \cite{ref35}, information elements capture the quintessence of \emph{information itself} in that all RVs within a given class can be derived from a representative RV for that class using finite state reversible encoding operations, i.e., with 1-to-1 mappings. Contrast the notion of information elements with the Shannon entropy of a source $X$, denoted $H(X)$. Two sources $X_1$ and $X_2$ might produce information at the same entropy rate\footnote{Most of the arguments here are valid for a countable $\mathcal{X}$. Entropies for countable alphabets can be infinite and even discontinuous. In the later sections, we shall be dealing solely with finite discrete RVs.}$^,$\footnote{For finite or countable $\mathcal{X}_1$, $\mathcal{X}_2$, if $f:\mathcal{X}_1\to\mathcal{X}_2$ is a bijection such that $X_2=f(X_1)$, then $H(X_2)=H(X_1)$, i.e., entropy is invariant under relabeling.}, but not necessarily produce the ``same'' information. Thus, $X_1 = X_2 \implies X_1 \sim^{S} X_2 \implies H(X_1)=H(X_2)$, but the converse of neither implication is true.

A partial order between two information elements $S_1$ and $S_2$ is defined as follows: $S_1 \succcurlyeq S_2$ iff $H({S_2}|{S_1})=0$ or equivalently iff $X_{S_2}$ is $\sigma(X_{S_1})$-measurable. We say that $S_1$ is \emph{larger} than $S_2$ or equivalently $S_2$ is an abstraction of $S_1$. Likewise, we write $S_1 \preccurlyeq S_2$ if $S_2 \succcurlyeq S_1$, when $S_1$ is \emph{smaller} than $S_2$. There exists a natural metric $\rho$ on the space of information elements and an associated topology induced by $\rho$ \cite{ref35}. $\rho$ is defined as follows: $\rho(S_1,S_2)=H({S_1}|{S_2})+H({S_2}|{S_1})$. Clearly, $\rho(S_1,S_2)=0$ iff $S_1 \succcurlyeq S_2$ and $S_2 \succcurlyeq S_1$. The \emph{join} of two information elements $S_1$ and $S_2$ is given by $\sup\{ {S_1},{S_2}\}$ (denoted $S_1\vee S_2$) and is called the \emph{joint information} of both $S_1$ and $S_2$. The joint RV $(X_{S_1},X_{S_2})$ is representative of the joint information. Likewise, the \emph{meet} is given by $\inf\{ {S_1},{S_2}\}$ (denoted $S_1\wedge S_2$) and is called the \emph{common information} of $S_1$ and $S_2$. $(X_{S_1}\wedge X_{S_2})$ is the representative \emph{common RV} \cite{ref35}. The entropy of both the joint and common information elements are invariant in a given equivalent class.

A finite set of information elements endowed with the partial order $\succcurlyeq$, join ($\vee$), and meet ($\wedge$) operations have the structure of a metric lattice which is isomorphic to a finite partition lattice \cite{ref35}, \cite{ref36}. As a simple example, the lattice structure arising out of a $\textsc{Xor}$ operation is the diamond lattice $M_3$, the smallest instance of a nondistributive modular lattice. The nondistributivity is easily seen as follows: let $S_3=\textsc{Xor}(S_1,S_2)$ where $S_1$ and $S_2$ are independent information elements.  In this example, $(S_3\wedge S_2)\vee (S_3\wedge S_1)=0$, whereas $S_3\wedge (S_2\vee S_1)=S_3\ne 0$. In general however, information lattices are neither distributive nor modular \cite{ref35}, \cite{ref36}. More important for our immediate purposes is the notion of common information as defined by Shannon \cite{ref35} which arises naturally when quantifying information embodied in structure. Contrast this with Shannon's mutual information which does not correspond to any element in the information lattice.

The modeling of information structures can also be motivated nonstochastically, i.e., when the underlying space has no probability measure associated with it (e.g., see \cite{ref48}, \cite{ref27}, \cite{ref38}, \cite{ref39}). Let $(\Omega, \mathfrak{F})$ be a measurable space, where $\Omega$ is the set of possible states of Nature, and elements of $\mathfrak{F}$ are events. One of the states $\omega \in \Omega$ is the ``true'' state. An event $E$ \emph{occurs} when $\omega \in E$. Define an \emph{uncertain variable} $X$ \cite{ref48} taking values in a discrete measurable space $(\mathcal{X},\mathfrak{X})$ as the measurable function $X:\Omega \to \mathcal{X}$ where $\mathfrak{X}$ contains all singletons. The $\sigma$-algebra induced by $X$ is $\sigma(X)=\sigma(\{X^{-1}(T):T \in \mathfrak{X}\})$. $X$ generates a partition on $\Omega$ called the \emph{information partition} $\mathcal{P}_X=\{X^{-1}(x) \in \Omega:x \in \mathcal{X}\}$. Since the alphabet $\mathcal{X}$ is finite or countable, $\sigma(\mathcal{P}_X)=\sigma(X)$.  

The information structure $\left\langle {\Omega,\mathcal{P}_X}\right\rangle$ specifies the extent to which an agent observing $X$ can distinguish among different states of Nature. Given an observation $x=X(\omega)$, an agent endowed with a partition $\mathcal{P}_X$ only knows that the true state belongs to $\mathcal{P}_X(\omega)$, where $\mathcal{P}_X(\omega)$ is the element of $X$'s partition that contains $\omega$. Given a pair of partitions $(\mathcal{P}_i,\mathcal{P}_j)$ on $\Omega$, $\mathcal{P}_i$ is said to be \emph{finer} than $\mathcal{P}_j$ and that $\mathcal{P}_j$ is \emph{coarser} than $\mathcal{P}_i$ if $\mathcal{P}_i(\omega)  \subseteq \mathcal{P}_j(\omega)$ $\forall \omega \in \Omega$. If $\mathcal{P}_i$ is finer than $\mathcal{P}_j$, then agent $i$ has \emph{more precise} information than agent $j$ in that $i$ can distinguish between more states of Nature. We say \emph{$X$ knows an event $E$ at $\omega$} if $\mathcal{P}_X(\omega) \subset E$. $E$ can only be known if it occurs. The event that \emph{$X$ knows $E$} is the set $K_X(E)=\{\omega: \mathcal{P}_X(\omega) \subset E\}$. Then, given two agents, Alice observing $X$ and Bob observing $Y$, $K_X(E) \cap K_Y(E)$ is the event that $E$ is \emph{mutually} known (between Alice and Bob). We say that an event $E$ is \emph{commonly} known (to both Alice and Bob) if it occurs, or equivalently, an \emph{event $E$ is common information} iff $E \in \sigma({\mathcal{P}_{\curlywedge}})$, where ${\mathcal{P}_{\curlywedge}}={{\mathcal{P}}_X}\wedge {{\mathcal{P}}_Y}$ is the finest common coarsening of the agents' partitions\footnote{The astute reader will immediately notice the connection with the notion of \emph{common knowledge} due to Aumann \cite{ref27}. In keeping with our focus on information structure, we prefer the term ``common information'' to ``common knowledge.'' Indeed, for finite or countably infinite information partitions, common knowledge is defined on the basis of the information contained in ${\mathcal{P}_{\curlywedge}}$ as follows. An event \emph{$E$ is common knowledge at $\omega$} iff ${\mathcal{P}_{\curlywedge}}(\omega) \subset E$, i.e., the event that $E$ is common knowledge is $C(E)=\{\omega: \mathcal{P}_{\curlywedge}(\omega) \subset E\}$. For any event $E$, $C(E)\subset E$. $E$ is \emph{common information} if $C(E)=E$ \cite{ref27}, \cite{ref49}.}$^{,}$\footnote{For uncountable alphabets, see e.g., \cite{ref49} for a more nuanced discussion on representing information structures using $\sigma$-algebras of events instead of partitions.}.
Since the $\sigma$-algebra generated by ${\mathcal{P}_{\curlywedge}}$ is simply $\sigma({{\mathcal{P}}_X}) \cap \sigma ({{\mathcal{P}}_Y})$, or equivalently, $\sigma (X) \cap \sigma (Y)$, $E$ is common information iff $E \in \sigma (X) \cap \sigma (Y)$. Commonly knowing $E$ is a far stronger requirement than mutually knowing $E$. For finite $\mathcal{X}$, $\mathcal{Y}$, the common information structure admits a representation as a graph $C_{XY}$ with the vertex set  ${{\mathcal{P}}_X}\vee {{\mathcal{P}}_Y}$ and an edge connecting two vertices if the corresponding atoms $v_i$ and $v_j$ are contained in a single atom of ${{\mathcal{P}}_X}$ or ${{\mathcal{P}}_Y}$ or of both. 
The connected components of $C_{XY}$ are in one-to-one correspondence with the atoms of ${\mathcal{P}_{\curlywedge}}$ \cite{ref48}.

\begin{example}
Let $\Omega=\{\omega_1,\omega_2,\omega_3,\omega_4\}$. Alice observes $X$ which generates the information partition ${\mathcal{P}_X} = {\omega _1}{\omega _4}|{\omega _2}|{\omega _3}$. Likewise, Bob observes $Y$ which induces the partition, ${\mathcal{P}_Y} = {\omega _1}{\omega _2}|{\omega _3}|{\omega _4}$. Let $\omega_2$ be the true state of Nature. Consider the event $E=\{\omega_1,\omega_2\}$. Both Alice and Bob know $E$ at $\omega_2$, since ${\mathcal{P}_X}(\omega_2)=\{\omega_2\} \subset E$ and ${\mathcal{P}_Y}(\omega_2)=\{\omega_1,\omega_2\} \subset E$. The event that Alice knows $E$ is simply the true state $\{\omega_2\}$ \emph{(}i.e., $K_X(E)=\{\omega_2\}$\emph{)}, whereas for Bob, $K_Y(E)=\{\omega_1,\omega_2\}$. Clearly, Bob cannot tell apart the true state $\{\omega_2\}$ \emph{(}in which Alice knows $E$\emph{)} from $\{\omega_1\}$ \emph{(}in which Alice does not know $E$\emph{)}. Hence, $E$ is not commonly known to Alice and Bob. 

On the other hand, it is easy to check that the events $\{\omega_1,\omega_2,\omega_4\}$ and $\{\omega_3\}$ are common information. Indeed, ${\mathcal{P}_ \curlywedge } = {\mathcal{P}_X} \wedge {\mathcal{P}_Y} = \left\{ {\{ {\omega _1},{\omega _2},{\omega _4}\},\{ {\omega _3}\} } \right\}$. $C_{XY}$ has the vertex set ${{\mathcal{P}}_X}\vee {{\mathcal{P}}_Y}=\{\{\omega_1\},\{\omega_2\},\{\omega_3\},\{\omega_4\}\}$ and the connected components of $C_{XY}$ correspond to the atoms $\{\omega_1,\omega_2,\omega_4\}$ and $\{\omega_3\}$ of ${\mathcal{P}_{\curlywedge}}$.
\end{example}

One may also seek to characterize the private information structures of the agents. Let $\Omega$ be a finite set of states of Nature. To simplify notation, let $X$ denote the agent $X$ as well as its information partition. Let Alice and Bob be endowed, respectively, with information partitions $X$ and $Y$ so that $X$ and $Y$ are subalgebras of a $2^{|\Omega|}$-element Boolean algebra. One plausible definition of the \emph{private information structure of $Y$} is the minimal amount of information that $X$ needs from $Y$ to reconstruct the joint information $Y\vee X$ \cite{ref37}. Define $P{{I}_{X}}(Y)=\{Z:Z\vee X=Y\vee X; Z\subseteq Y; Z \text{ minimal}\}$. Since $P{{I}_{X}}(Y)$ complements $X$ to reconstruct $Y\vee X$, minimality of $Z$ entails that $\forall Z \in P{{I}_{X}}(Y)$, $Z\wedge X=0$, where $0$ denotes the two-element algebra. Witsenhausen \cite{ref37} showed that the problem of constructing elements of $P{{I}_{X}}(Y)$ with minimal cardinality is equivalent to the chromatic number problem for a graph $G_Y$ with the vertex set $Y$ and an edge connecting vertices  $v_i$ and $v_j$ iff there exists an atom $x \in X$ such that $v_i \cap x \ne \emptyset$ and $v_j \cap x \ne \emptyset$. Unfortunately, since there are multiple valid minimal colorings of $G_Y$, $P{{I}_{X}}(Y)$ is not be unique. The following example illustrates the point. 

\begin{example}
Consider the set, $\Omega  = \left\{ {{\omega _1}, \ldots ,{\omega _{16}}} \right\}$. Let Alice and Bob's partitions be respectively, ${X}={\omega _1}{\omega _3}|{\omega _4}{\omega _5}|{\omega _6}{\omega _7}|{\omega _8}$ ${\omega _9}|{\omega _{10}}{\omega _2}|{\omega _{11}}{\omega _{13}}|{\omega _{14}}{\omega _{15}}|{\omega _{12}}{\omega _{16}}$ and ${Y} = {\omega _1}{\omega _2}|{\omega _3}{\omega _4}|{\omega _5}$ ${\omega _6}|{\omega _7}{\omega _8}|{\omega _9}{\omega _{10}}|{\omega _{11}}{\omega _{12}}|{\omega _{13}}{\omega _{14}}|{\omega _{15}}{\omega _{16}}$. $G_Y=(Y,\mathcal{E})$ has the edge set $\mathcal{E}=\{ \{ {\omega _1}{\omega _2},{\omega _3}{\omega _4}\} ,\{ {\omega _3}{\omega _4},{\omega _5}{\omega _6}\} ,\{ {\omega _5}{\omega _6},{\omega _7}{\omega _8}\}$, $\{ {\omega _7}{\omega _8},{\omega _9}{\omega _{10}}\} ,\{ {\omega _9}{\omega _{10}},{\omega _1}{\omega _2}\} ,\{ {\omega _{11}}{\omega _{12}},{\omega _{13}}{\omega _{14}}\},\{ {\omega _{13}}{\omega _{14}},{\omega _{15}}$ ${\omega _{16}}\},\{ {\omega _{15}}{\omega _{16}},{\omega _{11}}{\omega _{12}}\} \}$. 
 
Two distinct minimal colorings of $G_Y$ are as follows:
 \begin{align*}
 (a)\; {\gamma _1} &= \{ {\omega _1}{\omega _2}, {\omega _7}{\omega _8}, {\omega _{11}}{\omega _{12}}\},
 {\gamma _2} = \{ {\omega _3}{\omega _4},{\omega _9}{\omega _{10}},{\omega _{15}}{\omega _{16}}\},\\
 {\gamma _3} &= \{ {\omega _5}{\omega _6},{\omega _{13}}{\omega _{14}}\},
 \end{align*}
so that $$PI_X^a(Y) = {\omega _1}{\omega _2}{\omega _7}{\omega _8}{\omega _{11}}{\omega _{12}}|{\omega _3}{\omega _4}{\omega _9}{\omega _{10}}{\omega _{15}}{\omega _{16}}|{\omega _5}{\omega _6}{\omega _{13}}{\omega _{14}},$$
and 
\begin{align*}
 (b)\; {{\gamma'}_1} &= \{ {\omega _1}{\omega _2},{\omega _5}{\omega _6},{\omega _{11}}{\omega _{12}}\}, 
 {{\gamma'}_2} = \{ {\omega _3}{\omega _4},{\omega _7}{\omega _8},{\omega _{13}}{\omega _{14}}\},\\
 {{\gamma'}_3} &= \{ {\omega _9}{\omega _{10}},{\omega _{15}}{\omega _{16}}\},
 \end{align*}
 so that $$PI_X^b(Y) = {\omega _1}{\omega _2}{\omega _5}{\omega _6}{\omega _{11}}{\omega _{12}}|{\omega _3}{\omega _4}{\omega _7}{\omega _8}{\omega _{13}}{\omega _{14}}|{\omega _9}{\omega _{10}}{\omega _{15}}{\omega _{16}}.$$

It is easy to see that $PI_X^a(Y) \vee X = PI_X^b(Y) \vee X = Y \vee X$. Hence, such a minimal coloring is not unique and consequently, $P{I_X}(Y)$ is not unique.
\end{example}

One would also like to characterize the information contained exclusively in either $X$ or $Y$. The \emph{private information structure of $Y$ with respect to $X$} may be defined as the amount of information one needs to reconstruct $Y$ from the common information $X\wedge Y$. Define $PI(Y\backslash X)=\{Z:Z\vee(X\wedge Y)=Y;Z \text{ minimal}\}$, where minimality of $Z$ entails that $\forall Z \in PI(Y\backslash X)$, if there exists a $Z'$ such that $Z' \supseteq Z$ and $Z'\vee(X\wedge Y)=Y$, then $Z'\not\in PI(Y\backslash X)$. We note that, if $Z \in PI(Y\backslash X)$, then $Z\vee X=Y\vee X$ and $Z\wedge X=0$. Hexner and Ho \cite{ref38} proposed and showed that this definition does not admit a unique specification for the private information of $Y$ with respect to $X$ as can be seen from the following example. 

\begin{example}
Consider the set, $\Omega =\{\omega_1,\ldots,\omega_6\}$ and the following partitions on $\Omega:$ $X=\omega_1\omega_2|\omega_3|\omega_4\omega_5|\omega_6$, and $Y=\omega_1|\omega_2\omega_3|\omega_4|\omega_5\omega_6$. Then we have, $X \vee Y = \omega_1|\omega_2|\omega_3|\omega_4|\omega_5|\omega_6$ and $X \wedge Y = \omega_1\omega_2\omega_3|\omega_4\omega_5\omega_6$. It is easy to see that each of the following subalgebras satisfies the definition, i.e., given $Z_1=\omega_1\omega_4|\omega_2\omega_3\omega_5\omega_6$ and $Z_2 = \omega_1\omega_5\omega_6|\omega_2\omega_3\omega_4$, we have, $Z_1 \vee (X \wedge Y) = Z_2 \vee (X \wedge Y) = Y$ and $Z_1 \vee X = Z_2 \vee X = Y \vee X$. Hence, $PI(Y\backslash X)$ is not unique.
\end{example}

\begin{remark}
\emph{We have the following observations. Note that if $Z_1 \in PI(Y\backslash X)$, then $Z_1 \vee X = Y \vee X$. Thus, one can find a $Z_2 \in P{{I}_{X}}(Y)$ such that $Z_2 \subseteq Z_1$. Choosing $Z_1$ minimal, it follows that the cardinality of the minimal algebras of $PI(Y\backslash X)$ is lower bounded by the cardinality of the minimal algebras of $P{{I}_{X}}(Y)$ or equivalently by the chromatic number of $G_Y$. Thus, $X$ need not use all of $PI(Y\backslash X)$ to reconstruct $Y\vee X$. Furthermore, it is known that the lattice $L$ of subalgebras of a finite Boolean algebra is isomorphic to a finite partition lattice \cite{ref44}. Thus, in general, $L$ is not distributive, nor even modular. Since both the structures $P{{I}_{X}}(Y)$ and $PI(Y\backslash X)$ consists of complements in $L$, nonmodularity of $L$ implies the nonuniqueness of the private information structures.}
\end{remark}

\subsection{Operational Aspects}
    
We now turn to mainstream information-theoretic notions of ``common information'' (CI). We introduce the remaining notation. 
For a discrete, finite-valued RV $X$, ${{p}_{X}}(x)=\mathbb{P}\{X=x\}$ denotes the probability mass function (pmf or distribution) of $X$. We abbreviate ${{p}_{X}}(x)$ as $p(x)$ when there is no ambiguity. For $\mathcal{X}=\{x_n,\;n=1,\ldots,N\}$, the entropy $H(X)$ of $X$ can be written as $H({{p}_{1}},\ldots,{{p}_{N}})\defeq\sum\nolimits_{n}{{{p}_{n}}\log \tfrac{1}{{{p}_{n}}}}$, where $p_n=\mathbb{P}\{X=x_n\}$ and $\sum\nolimits_{n}{{{p}_{n}}}=1$.
The Kullback-Leibler (KL) divergence from ${{q}_{X}}$ to ${{p}_{X}}$ is defined as $D(p||q)\defeq{{\sum }_{x\in \mathcal{X}}}{{p}_{X}}(x)\log \tfrac{{{p}_{X}}(x)}{{{q}_{X}}(x)}$. 

$X-Y-Z$ denotes that $X$ is conditionally independent of $Z$ given $Y$ (denoted $X\perp Z|Y$), or equivalently, $X,Y,Z$ form a Markov chain satisfying 
$$p(x,y,z)=\tfrac{p(x,y)p(y,z)}{p(y)}=p(x|y)p(y,z),\text{ if }p(y)>0;\text{ else } 0.$$
Equivalently, $p(y)p(x,y,z)=p(x,y)p(y,z)$. 

Let $\{{{X}_{i}},{{Y}_{i}}\}_{i=1}^{\infty }$ be i.i.d. copies of the pair $(X,Y)\sim {{p}_{XY}}$ on $\mathsf{\mathcal{X}}\times \mathsf{\mathcal{Y}}$. An information source generating such a (stationary) sequence is called a two-component discrete, memoryless source (2-DMS). Given $\varepsilon>0$, we say that $\hat{X}^n$ $\varepsilon$-recovers $X^n$ iff $\mathbb{P}\{\hat{X}^n \ne X^n\} < \varepsilon$.

To fix ideas, consider a ``one-decoder'' network for the distributed compression of a 2-DMS \cite{ref32l}. The correlated streams $\{{{X}_{i}}\}_{i=1}^{\infty }$ and $\{{{Y}_{i}}\}_{i=1}^{\infty }$ are encoded separately at rates $R_x$ and $R_y$ and decoded jointly by combining the two streams to $\varepsilon$-recover $(X^n,Y^n)$. A remarkable consequence of the Slepian-Wolf theorem \cite{ref32l} is that the (minimum) sum rate of $R_x+R_y=H(X,Y)$ is achievable. This immediately gives a coding-theoretic interpretation of Shannon's mutual information (MI) as the maximum descriptive savings in sum rate by considering $(X,Y)$ jointly rather than separately, i.e., 
$$I(X;Y)=H(X)+H(Y)-\min (R_x+R_y).$$
Thus, for the one-decoder network, MI appears to be a natural measure of CI of two dependent RVs. However, other networks yield different CI measures. Indeed, as pointed out in \cite{ref33}, depending upon the number of encoders and decoders and the network used for connecting them, several notions of CI can be defined. We restrict ourselves to two dependent sources and a ``two-decoder'' network when two different notions of CI due to G{\'a}cs and K\"{o}rner \cite{ref31} and Wyner \cite{ref32} are well known. Each of these notions appear as solutions to asymptotic formulations of some distributed information processing task.

Given a sequence $(X^n,Y^n)$ generated by a 2-DMS $(\mathcal{X}\times \mathcal{Y},\text{ }p_{XY})$, G{\'a}cs and K\"{o}rner (GK) \cite{ref31}
defined CI as the maximum rate of common randomness (CR) that two nodes, observing sequences ${{X}^{n}}$ and ${{Y}^{n}}$ separately can extract without any communication, i.e., 
\begin{align*}
C_{GK}(X;Y)\defeq\sup \tfrac{1}{n}H({{f}_{1}}({{X}^{n}})),
\end{align*}
where the supremum is taken over all sequences of pairs of deterministic mappings $(f_{1}^{n},f_{2}^{n})$ such that $\mathbb{P}\{f_{1}^{n}({{X}^{n}})\ne f_{2}^{n}({{Y}^{n}})\}\to 0\text{ as }n\to \infty$. 

The zero pattern of ${{p}_{XY}}$ is specified by its characteristic bipartite graph ${{B}_{XY}}$ with the vertex set $\mathcal{X}\cup \mathcal{Y}$ and an edge connecting two vertices $x$ and $y$ if ${{p}_{XY}} > 0$. If ${{B}_{XY}}$ is a single connected component, we say that ${{p}_{XY}}$ is \emph{indecomposable}. An \emph{ergodic decomposition} of ${{p}_{XY}}$ is defined by a unique partition of the space $\mathcal{X}\times \mathcal{Y}$ into connected components \cite{ref31}, \cite{ref33}, \cite{ref25}. Given an ergodic decomposition of $p_{XY}$ such that $\mathcal{X}\times\mathcal{Y}=\bigcup\nolimits_{{q_*}} \mathcal{X}_{q_*}\times\mathcal{Y}_{q_*}$, define the RV $Q_*$ as $Q_*=q_* \iff X \in {\mathcal{X}}_{q_*} \iff Y \in {\mathcal{Y}}_{q_*}$. For any RV $Q$ such that $H(Q|X) = H(Q|Y) = 0$, we have $H(Q|Q_*) = 0$ so that $Q_*$ has the maximum range among all $Q$ satisfying $H(Q|X) = H(Q|Y) = 0$. In this sense, $Q_*$ is the maximal common RV\footnote{It is not hard to see the connection with Shannon's notion of common information introduced earlier in Section II.A. In particular, we have $Q_*=X\wedge Y$. G{\'a}cs and K\"{o}rner independently proposed the notion of common information two decades following Shannon's work \cite{ref35}.} of $X$ and $Y$. Remarkably, GK showed that
\begin{align*}
C_{GK}(X;Y) = H(Q_*) \label{eq:1} \tag{1}    
\end{align*}
Thus, common GK codes cannot exploit any correlation beyond deterministic interdependence of the sources. $C_{GK}(X;Y)$ depends solely on the zero pattern of $p_{XY}$ and is zero for all indecomposable distributions. 

The following double markovity lemma (see proof in Appendix A) is useful.
\begin{lemma}
A triple of RVs $(X,Y,Q)$ satisfies the double Markov conditions 
\begin{align*}
X-Y-Q,\text{ }Y-X-Q \tag{2}
\end{align*}
iff there exists a pmf $p_{Q'|XY}$ such that $H(Q'|X)=H(Q'|Y)=0$ and $XY-Q'-Q$. Furthermore, (2) implies $I(XY;Q) = H(Q')$ iff $H(Q'|Q)=0$.
\end{lemma}

\begin{remark}
\emph{For all $X,Y$ we have $I(X;Y)=H({{Q}_{*}})+I(X;Y|{{Q}_{*}})$. We say that ${{p}_{XY}}$ is \emph{saturable} if $I(X;Y|{{Q}_{*}})=0$. Equivalently, ${{p}_{XY}}$ is saturable iff there exists a pmf ${{p}_{Q|XY}}$ such that $X-Q-Y,\text{  }Q-X-Y,\text{  }Q-Y-X$ (see Lemma A1 in Appendix A). We say that the triple $(X,Y,Q)$ has a \emph{pairwise double Markov} structure when the latter condition holds.}
\end{remark}

The following alternative characterizations of $C_{GK}(X;Y)$ follow from Lemma 2 \cite{ref33}.
\begin{align*}
C_{GK}(X;Y) &= \mathop {\max }\limits_{\substack{ Q:\text{ }Q-X-Y\\ \hspace{4mm} Q-Y-X}} I(XY;Q)\\
&= I(X;Y)-\mathop {\min }\limits_{\substack{ Q:\text{ }Q-X-Y\\ \hspace{4mm} Q-Y-X}} I(X;Y|Q),\tag{3}
\end{align*}
where the cardinality of the alphabet $\mathcal{Q}$ is bounded as $|{\mathcal{Q}}|\leq |\mathcal{X}||\mathcal{Y}|+2$.

Wyner \cite{ref32} defined CI as the minimum rate of CR needed to simulate a 2-DMS $(\mathcal{X}\times \mathcal{Y}, \; \; p_{XY})$ using local operations and no communication. More precisely, given access to a common uniform random string $Q_n\sim \operatorname{unif}[1:2^{nR}]$ and independent noisy channels $p_{{\hat{X}^{n}}|Q_n}(x^n|q)$ and $p_{{\hat{Y}^{n}}|Q_n}(y^n|q)$ such that $\left({\hat{X}^{n}},{\hat{Y}^{n}}\right)$ $\varepsilon$-recovers $(X^n,Y^n)$, the Wyner CI, denoted $C_W(X;Y)$, is the minimum cost (in terms of the number of common random bits per symbol $R$) for the distributed approximate simulation of ${{p}_{XY}}$. $C_W(X;Y)$ admits an elegant single-letter characterization,
\begin{align*}
  C_W(X;Y) &\defeq \mathop {\min }\limits_{Q:X - Q - Y} I(XY;Q)\\
  &= I(X;Y)+\mathop {\min }\limits_{Q:X - Q - Y} I(Y;Q|X)+I(X;Q|Y), \tag{4}
\end{align*}
where again $|{\mathcal{Q}}|\leq |\mathcal{X}||\mathcal{Y}|+2$.

A related notion of \emph{common entropy}, $G(X,Y)$ is useful for characterizing a zero-error version of the Wyner CI \cite{ref32n}. 
\begin{align*}
G(X;Y) \defeq \mathop {\min }\limits_{Q:X - Q - Y} H(Q)  \tag{5}
\end{align*}

Gray and Wyner (GW) \cite{ref32m} devised a distributed lossless source coding network for jointly encoding the 2-DMS into a common part (at rate $R_c$) and two private parts (at rates $R_x$ and $R_y$), and separately decoding each private part using the common part as side information. The optimal rate region ${\Re _{\op{GW}}}(X;Y)$ for this ``two-decoder'' network configuration is given by,
\begin{displaymath}
{\Re _{\op{GW}}}(X;Y) = \begin{cases}
(R_c,R_x,R_y) \in \mathbb{R}_ + ^{3}:\exists {p_{Q|XY}} \in {{{\mathcal{P}}}_{XY}},\\
\text{s.t. } {R_c} \geq I({XY};Q), \\ 
\hspace{5.5mm} {R_x} \geq H({X}|Q),{\text{ }}{R_y} \geq H({Y}|Q),
\end{cases}
\end{displaymath}
where ${{\mathcal{P}}_{XY}}$ is the set of all conditional pmfs ${p_{Q|XY}}$ s.t. $|{\mathcal{Q}}|\leq |\mathcal{X}||\mathcal{Y}|+2$. A trivial lower bound to ${\Re _{\op{GW}}}(X;Y)$ follows from basic information-theoretic considerations \cite{ref32m},
\begin{displaymath}
\begin{gathered}
  {\Re _{\op{GW}}}(X;Y) \subseteq {\mathfrak{L}_{\op{GW}}}(X;Y) \hfill \\
  \hspace{12mm} = \left\{ \begin{gathered}
  ({R_c},R_x,R_y):{R_c} + {R_x} \geq H({X}), \hfill \\
  \hspace{20.5mm} {R_c} + {R_y} \geq H({Y}), \hfill \\
  \hspace{20.5mm} {R_0} + R_x + R_y  \geq H(XY) \hfill 
\end{gathered}  \right\}. \hfill
\end{gathered} 
\end{displaymath}
The different notions of CI can be viewed as extreme points for the corresponding common rate $R_c$ in the two-decoder network\footnote{See Problem 16.28--16.30, p. 394 in \cite{ref34}}, i.e., for $({R_x},{R_y},{R_c}) \in {\Re _{\op{GW}}}(X;Y)$, we have
\begin{align*}
C_{GK}(X;Y) &=\mathop {\max }\limits_{ \hspace{4mm} {R_c} + {R_x} = H(X),{\text{ }}{R_c} + {R_y} = H(Y)} R_c, \\
I(X;Y) &= \mathop {\max }\limits_{\hspace{4mm} 2{R_c} + {R_x} + {R_y} = H(X) + H(Y)} R_c, \\
C_W(X;Y) &=\mathop {\min }\limits_{ \hspace{4mm} {R_c} + {R_x} + {R_y} = H(X,Y)} R_c.
\end{align*}

\begin{remark}
\emph{The different notions of CI are related as, ${{C}_{GK}}(X;Y)\leq I(X;Y)\leq {{C}_{W}}(X;Y)$, with equality iff $p_{XY}$ is saturable, whence ${{C}_{GK}}(X;Y)=I(X;Y)\iff I(X;Y)=C_W(X;Y)$ (see Lemma A2 in Appendix A).}
\end{remark}

\begin{remark}
\emph{$C_{GK}(X_1;\ldots;X_K)$ is monotonically nonincreasing in the number of input arguments $K$. In contrast, $C_W(X_1;\ldots;X_K)$ is monotonically nondecreasing in $K$. It is easy to show that $C_{GK}(X_1;\ldots;X_K) \leq \mathop{\min}\limits_{i\ne j}\; I(X_i;X_j)$, while $C_W(X_1;\ldots;X_K) \geq \mathop{\max}\limits_{i\ne j}\; I(X_i;X_j)$ for any $i,j\in \{1,\ldots,K\}$ (see Lemma A3 in Appendix A).}
\end{remark}

Witsenhausen \cite{ref41} defined a symmetric notion of private information. Witsenhausen?s total private information, denoted ${{M}_{W}}(X;Y)$, is defined as the complement of Wyner's CI,
\begin{align*}
M_W(X;Y) \defeq H(XY)-C_W(X;Y)=\mathop {\max }\limits_{Q:{\text{ }}X - Q - Y} H(XY|Q).
\end{align*}

One can define the private information of $Y$ with respect to $X$ (denoted ${\tilde P_W}(Y\backslash X)$) as
\begin{align*}
{\tilde P_W}(Y\backslash X) \defeq \mathop {\max }\limits_{\substack{ Q:\text{ }X-Q-Y\\ \hspace{4mm} X-Y-Q}} H(Y|Q).\tag{6}
\end{align*}
Likewise, the complement of ${\tilde P_W}(Y\backslash X)$ is defined as
\begin{align*}
{\tilde C_W}(Y\backslash X) \defeq \mathop {\min }\limits_{\substack{ Q:\text{ }X-Q-Y\\ \hspace{4mm} X-Y-Q}} H(Q). \tag{7}
\end{align*}
The double Markov constraint (see Lemma 2) already hints at the structure of the minimizer $Q$ in (7). The following lemma (see proof in Appendix A) shows that the minimizer in ${{\tilde{C}}_{W}}(Y\backslash X)$ is a minimal sufficient statistic of $Y$ with respect to $X$.

\begin{lemma}
Let $Q_{Y}^{X}$ denote a function $f$ from $\mathsf{\mathcal{Y}}$ to the probability simplex ${{\Delta }_{\mathsf{\mathcal{X}}}}$ (the space of all distributions on $\mathcal{X}$) that defines an equivalence relation on $\mathcal{Y}$:
\begin{align*}
y \equiv y'{\text{ }}\operatorname{iff}{\text{ }}{p_{X|Y}}(x|y) = {p_{X|Y}}(x|y'),{\text{ }}x \in \mathcal{X},{\text{ }}y,y' \in \mathcal{Y}.
\end{align*}
Then $Q_{Y}^{X}$ is a minimal sufficient statistic of $Y$ with respect to $X$.
\end{lemma}

Theorem 1 gives a decomposition of $H(Y)$ into a part that is correlated with $X$ ($H(Q_{Y}^{X})$) and a part that carries no information about $X$ ($H(Y|Q_{Y}^{X})$) (see proof in Appendix A). 

\begin{theorem}
For any pair of correlated RVs $(X,Y)\sim p_{XY}$, the following hold:
\begin{align*}
  &{{\tilde C}_W}(Y\backslash X)= H(Q_Y^X), \tag{8a} \\
  &{{\tilde P}_W}(Y\backslash X)= H(Y|Q_Y^X), \tag{8b} \\
  &H(Y)= {{\tilde C}_W}(Y\backslash X)+{{\tilde P}_W}(Y\backslash X) = H(Q_Y^X)+H(Y|Q_Y^X), \tag{8c}\\
  &{C_W}(X;Y)\leq {{\tilde C}_W}(Y\backslash X). \tag{8d}
\end{align*}
\end{theorem}

Let ${\mathcal{X}}_k \subseteq \mathcal{X}$, ${\mathcal{Y}}_k \subseteq \mathcal{Y}$, where $\mathcal{X}_{k}$$'$s and $\mathcal{Y}_{k}$$'$s having different subscripts are distinct (but not necessarily disjoint) subsets. 
Let $({\mathcal{X}},{\mathcal{Y}})$ admit a unique decomposition into \emph{components} $\{({\mathcal{X}}_k,{\mathcal{Y}}_k)\}_{k=1}^{p}$ so that $\bigcup\nolimits_{k=1}^p {\mathcal{X}}_k=\mathcal{X}$, and $\{{\mathcal{Y}}_k\}_{k=1}^{p}$ is a partition of $\mathcal{Y}$ induced by the equivalence relation in Lemma 3, i.e., $\forall y,{y}'\in {{{\mathcal{Y}}}_{k}},\text{ }x\in {{{\mathcal{X}}}_{k}},\text{ }y\equiv {y}'$ and $\forall y\in {{{\mathcal{Y}}}_{k}},\text{ }x\notin {{{\mathcal{X}}}_{k}},\text{ }{{p}_{Y|X}}(y|x)=0$. 
We also require that each component is the ``largest'' possible in the sense that for any two components $({\mathcal{X}}_i,{\mathcal{Y}}_i)$, $({\mathcal{X}}_j,{\mathcal{Y}}_j)$, there exists $x' \in {\mathcal{X}}_i \cup {\mathcal{X}}_j$ such that ${p_{X|Y}}(x'|y_i)\ne{p_{X|Y}}(x'|y_j)$.
The size of the component $({{{\mathcal{X}}}_{k}},{{{\mathcal{Y}}}_{k}})$ is defined as $|{{{\mathcal{Y}}}_{k}}|$. 
Given such a unique decomposition of $({\mathcal{X}},{\mathcal{Y}})$ into components $\{({\mathcal{X}}_k,{\mathcal{Y}}_k)\}_{k=1}^{p}$, the following theorem gives necessary and sufficient conditions for ${{\tilde{P}}_{W}}(Y\backslash X)$ achieving its minimum and maximum value (see proof in Appendix A).

\begin{theorem}
${{\tilde{P}}_{W}}(Y\backslash X)$ achieves its minimum, ${{\tilde{P}}_{W}}(Y\backslash X)=0$ iff there exist no component with size greater than one. 

On the other hand, ${{\tilde{P}}_{W}}(Y\backslash X)$ achieves its maximum, ${{\tilde{P}}_{W}}(Y\backslash X)=H(Y|X)$ iff ${{p}_{XY}}$ is saturable iff each component $({{{\mathcal{X}}}_{k}},{{{\mathcal{Y}}}_{k}})$ is a connected component induced by the ergodic decomposition of ${{p}_{XY}}$.
\end{theorem}

\begin{example}
${{\tilde{P}}_{W}}(Y\backslash X)$ attains the lower bound for the following distribution $p_{XY}^{{}}$. Let $\mathcal{X}=\{1,2,3,4\},\mathcal{Y}=\{5,6,7\}$. We write $p_{XY}^{{}}(a,b)=(ab)$. Given, $(15)=\tfrac{5}{16},(17)=\tfrac{1}{8},(25)=\tfrac{3}{32},(27)=\tfrac{7}{32},(35)=\tfrac{5}{32},(37)=\tfrac{1}{16},(46)=\tfrac{1}{32}$, or graphically, ${{p}_{XY}}=\tfrac{1}{32}\left( \begin{matrix} 10 & 3 & 5 & .  \\    . & . & . & 1  \\    4 & 7 & 2 & .  \\ \end{matrix} \right)$. Let $f(y)={{p}_{X|Y=y}}$. Then we have $f(5)=[\tfrac{5}{9},\tfrac{1}{6},\tfrac{5}{18},0]$, $f(6)=[0,0,0,1]$, and $f(7)=[\tfrac{4}{13},\tfrac{7}{13},\tfrac{2}{13},0]$, so that $H(Q_{Y}^{X})=H(f(Y))=H(\tfrac{9}{16},\tfrac{1}{32},\tfrac{13}{32})=H(Y)=1.15$. Consequently ${{\tilde{P}}_{W}}(Y\backslash X)=0$. One can also easily verify that $H(Q_{X}^{Y})=H(\tfrac{21}{32},\tfrac{5}{16},\tfrac{1}{32})<H(X)$.
\end{example}

The quantity ${{\tilde{C}}_{W}}(Y\backslash X)$ first appeared in \cite{ref25} where it was called the dependent part of $Y$ from $X$.
Intuitively, ${{\tilde{C}}_{W}}(Y\backslash X)$ is the rate of the information contained in $Y$ about $X$. 
${{\tilde{C}}_{W}}(Y\backslash X)$ also appears in \cite{ref42}, \cite{ref43} and has the following coding-theoretic interpretation in a  source network with coded side information setup\footnote{See Theorem 16.4, p. 361 and Problem 16.26, p. 393 in \cite{ref34}.} where $X$ and $Y$ are encoded independently (at rates $R_X$ and $R_Y$, resp.) and a (joint) decoder needs to recover $X$ (with small error probability) using the rate-limited side information $Y$: ${{\tilde{C}}_{W}}(Y\backslash X)$ is the minimum rate $R_Y$ such that $R_X=H(X|Y)$ is achievable \cite{ref42}.
The following example shows that even though $H(Y)$ admits a decomposition\footnote{From an information structure aspect, recall that owing to the nonmodularity of the information lattice, even a unique decomposition into private and common information structures is not guaranteed (see Remark 1).} of the form in (8c), it might not always be possible to isolate its parts \cite{ref42}.

\begin{example}
Let $\mathcal{X}=\{1,2\}$ and $\mathcal{Y}=\{3,4,5,6\}$. Consider the perturbed uniform distribution $p_{XY}$ with $(13)=(14)=(15)=(16)=\tfrac{1}{8},(23)=\tfrac{1}{8}-\delta,(24)=\tfrac{1}{8}+\delta,(25)=\tfrac{1}{8}+{\delta }',(26)=\tfrac{1}{8}-{\delta }'$, where $\delta ,{\delta }'<\tfrac{1}{8}$. If $\delta ={\delta }'=\tfrac{1}{16}$, $H(Q_{Y}^{X})=H(\tfrac{3}{8},\tfrac{5}{8})<H(Y)$. However, if $\delta \ne {\delta }'$, then $H(Q_{Y}^{X})=H(Y)$. In fact, if $\delta \ne {\delta }'$, as $\delta,{\delta }'\to 0$, $H(Q_{Y}^{X})=H(Y)\approx 2$, while $I(X;Y)\to 0$. 
Thus, even when $I(X;Y)\ll H(Y)$, one needs to transmit the entire $Y$ (i.e., $R_Y \geq H(Y)$) to convey the full information contained in $Y$ about $X$.
\end{example}

\subsection{Related Common Information Measures}

We now briefly review some related candidate bivariate correlation measures. We highlight a duality in the optimizations in computing the various CI quantities. 

Starting with Witsenhausen \cite{ref41m}, the Hirschfeld-Gebelein-R{\'e}nyi (HGR) maximal correlation \cite{ref41n} has been used to obtain many impossibility results for the noninteractive simulation of joint distributions \cite{ref64}. The maximal correlation, denoted $\hgr{X;Y}$, is a function of $p_{XY}(x,y)$ and is defined as
\begin{align*}
\hgr{X;Y}=\mathbb{E}[f_1(X)f_2(Y)]
\end{align*}
where $\mathbb{E}[\cdot]$ is the expectation operator and the supremum is taken over all real-valued RVs $f_1(X)$ and $f_2(Y)$ such that $\mathbb{E}[f_1(X)]=\mathbb{E}[f_2(Y)]=0$ and $\mathbb{E}[f_1^2(X)]=\mathbb{E}[f_2^2(Y)]=1$. $\hgr{X;Y}$ has the following geometric interpretation \cite{ref41m}: if $L^2(X,Y)$ is a real separable Hilbert space, then $\hgr{X;Y}$ measures the cosine of the angle between the subspaces $L^2(X)=\{f_1(X):\mathbb{E}[f_1]=0,\textbf{ }\mathbb{E}[f_1^2]<\infty\}$ and $L^2(Y)=\{f_2(Y):\mathbb{E}[f_2]=0,\textbf{ }\mathbb{E}[f_2^2]<\infty\}$. $\hgr{X;Y}$ shares a number of interesting properties with $I(X;Y)$, viz., (a) nonnegativity: $0\leq \hgr{X;Y}\leq 1$ with $\hgr{X;Y}=0$ iff $X\perp Y$, and $\hgr{X;Y}=1$ iff $C_{GK}(X;Y)>0$, i.e. iff $p_{XY}(x,y)$ is decomposable \cite{ref41m}, and (b) data processing: $X'-X-Y-Y'\implies (\hgr{X';Y'}\leq\hgr{X;Y})$. 

Intuitively, for indecomposable distributions, if $\hgr{X;Y}$ is near 1, then $(X,Y)$ have still lots in common. 
Consider again the GK setup with node $\mathcal{X}$ observing $X^n$, node $\mathcal{Y}$ observing $Y^n$, where $(X^n,Y^n)$ is generated by a 2-DMS $(\mathcal{X}\times \mathcal{Y},\text{ }p_{XY})$.
Now, a (one-way) rate-limited channel is made available from node $\mathcal{Y}$ to $\mathcal{X}$. Then per \cite{ref65}, the maximum rate of CR extraction at rate $R$ (denoted $C(R)$) is, 
\begin{align*}
C(R)=\mathop {\max }\limits_{\substack{p_{Q|Y}:\text{ }I(Q;Y)-I(Q;X)\leq R}} I(Q;Y).
\end{align*}
We have $C_{GK}(X;Y)=C(0)$ by definition. Hence, if $R=0$, for indecomposable sources, not even a single bit in common can be extracted \cite{ref41m}. 
But if $R>0$, the first few bits of communication can ``unlock'' the common core of the 2-DMS. Assuming $C(0)=0$, the initial efficiency of CR extraction is given by \cite{ref67}
\begin{align*}
C'(0)=\lim_{R\downarrow 0} \frac{C(R)}{R}=\frac{1}{1-(s^{*}(X;Y))^2},
\end{align*}
where $s^{*}(X;Y)=\mathop {\max}\limits_{\substack{p_{Q|Y}:\text{ }I(Q;Y)>0}} \tfrac{I(Q;X)}{I(Q;Y)}$.

Alternatively, given a 2-DMS $(\mathcal{X}\times \mathcal{Y},\text{ }p_{XY})$, one can define the maximum amount of information that a rate $R$ description of source $\mathcal{Y}$ conveys about source $\mathcal{X}$, denoted $\Upsilon(R)$, that admits the following single-letter characterization \cite{ref67}.
\begin{align*}
\Upsilon(R)=\mathop {\max }\limits_{\substack{p_{Q|Y}:\text{ }I(Q;Y)\leq R}} I(Q;X), \tag{9}
\end{align*}
where it suffices to restrict ourselves to $p_{Q|Y}$ with alphabet $\mathcal{Q}$ such that $|\mathcal{Q}|\leq |\mathcal{Y}|+1$.
The initial efficiency of information extraction from source $\mathcal{Y}$ is given by
\begin{align*}
\Upsilon'(0)=\left. \frac{d\Upsilon(R)}{dR} \right\vert_{R\downarrow 0}=s^*(X;Y).
\end{align*}
We have $s^{*}(X;Y)=1$ iff $C_{GK}(X;Y)>0$ \cite{ref67}. 

Interestingly, a dual of the optimization in (9) gives the well-known \emph{information bottleneck} (IB) optimization \cite{ref62} that provides a tractable algorithm for approximating the minimal sufficient statistic of $Y$ with respect to $X$ ($Q_{Y}^{X}$ in Lemma 3). For some constant $\epsilon$, the IB solves the nonconvex optimization problem,
\begin{align*}
\min_{\substack{p_{Q|Y}:\text{ }I(Q;X)\geq \epsilon}} I(Q;Y) \tag{10}
\end{align*}
by alternating iterations amongst a set of convex distributions \cite{ref62}.

Since ${C_W}(X;Y)$ is neither concave nor convex in $Q$, computation of $C_W(X;Y)$ remains a difficult extremization problem in general, and simple solutions exist only for some special distributions \cite{ref41}.

\subsection{New Measures}
A symmetric measure of CI that combines features of both the GK and Wyner measures can be defined by a RV $Q$ as follows.
\begin{align*}
C^1(X;Y)=\mathop {\min }\limits_{p_{Q|XY}} I(Y;Q|X)+I(X;Q|Y)+I(X;Y|Q), \tag{11} 
\end{align*}
where it suffices to minimize over all $Q$ such that $|{\mathcal{Q}}|\leq |\mathcal{X}||\mathcal{Y}|+2$. Observe that $C^1(X;Y)=0$ if $p_{XY}$ is saturable. $C^1(X;Y)$ thus quantifies the minimum distance to saturability. However, $C^1(X;Y)$ is much harder to compute than the GK CI. 

More useful for our immediate purposes is the following asymmetric notion of CI for 3 RVs $(X_1,X_2,Y)$ \cite{ref8}. 
\begin{align*}
C^2(\{X_1,X_2\};Y)=\mathop {\max }\limits_{Q:\text{ }Q-X_i-Y,\text{ }i=1,2} I(Q;Y) \tag{12} 
\end{align*}
It is easy to see that $C^2$ retains an important monotonicity property of the original definition of GK (see Remark 4) in that $C^2$ is  monotonically nonincreasing in the number of input $X_i$'s, i.e., $C^2(\{X_1,\ldots,X_K\};Y)\leq C^2(\{X_1,\ldots,X_{K-1}\};Y)$.

One can also define the following generalization of the Wyner common entropy in (5).
\begin{align*}
C^3(\{X_1,X_2\};Y) = \mathop {\min }\limits_{ Q:\text{ }X_i-Q-Y,\text{ }i=1,2} H(Q) \tag{13}
\end{align*}
It is easy to see that $C^3(\{X_1,X_2\};Y)\geq C^3(\{X_1\};Y)=G(X_1;Y) \geq C_{W}(X_1;Y)\geq I(X_1;Y)$. $C^3$ is monotonically nondecreasing in the number of input $X_i$'s.

Any reasonable CI-based measure of redundancy in the PI decomposition framework must be nonincreasing in the number of predictors. In the next section, we exclusively concentrate on $C^2$. Better understanding of $C^2$ will guide our investigation in Section III in search of an ideal measure of redundancy for PI decomposition.

\section{Partial Information Decomposition: The Case for One Target and Two Predictor Variables}
\vspace{2mm}

Consider the following generalization of Shannon's MI for three RVs $(X_1,X_2,Y)$, called \emph{co-information} \cite{ref46} or \emph{interaction information} (with a change of sign) \cite{ref46n}. 
\begin{align*}
{I_{Co}}({X_1};{X_2};{Y}) = I({X_1};{X_2})-I({X_1};{X_2}|{Y}) \tag{14}
\end{align*}
Co-information is symmetric with respect to permutations of its input arguments and can be interpreted as the gain or loss in correlation between two RVs, when an additional RV is considered. The symmetry is evident from noting that $I({X_1};{X_2})-I({X_1};{X_2}|{Y})=I({X_1};{Y})-I({X_1};{Y}|{X_2})=I({X_2};{Y})-I({X_2};{Y}|{X_1})$. Given a ground set $\Omega$ of RVs, the Shannon entropies form a Boolean lattice consisting of all subsets of $\Omega$, ordered according to set inclusions \cite{ref45}. Co-informations and entropies are M\"{o}bius transform pairs with the co-informations also forming a lattice \cite{ref46}. Co-information can however be negative when there is pairwise independence, as is exemplified by a simple two-input $\textsc{Xor}$ function, $Y=\textsc{Xor}({{X}_{1}};{{X}_{2}})$. Bringing in additional side information $Y$ induces artificial correlation between ${{X}_{1}}$ and ${{X}_{2}}$ when there was none to start with. Intuitively, these artificial correlations are the source of synergy. Indeed, co-information is widely used as a synergy-redundancy measure with positive values implying redundancy and negative values expressing synergy \cite{ref15,ref28,ref29,ref30}. However, as the following example shows, co-information confounds synergy and redundancy and is identically zero if the interactions induce synergy and redundancy in equal measure.

\begin{example}
Let $\mathcal{X}_1=\mathcal{X}_2=\mathcal{Y}=\{1,2,3,4\}$. We write $p_{{X_1}{X_2}{Y}}(a,b,c)\defeq(abc)$. Consider the following distribution: $(111)=(122)=(212)=(221)=(333)=(344)=(434)=(443)=\tfrac{1}{8}$. First note that $I(X_1X_2;Y)=2$ bits. The construction $p_{{X_1}{X_2}{Y}}$ is such that one bit of information about $Y$ is contained identically in both $X_1$ and $X_2$. The other bit of information about $Y$ is contained only in the joint RV $X_1X_2$. Thus, $X_1,X_2$ contains equal amounts of synergistic and redundant information about $Y$. However, it is easy to check that ${I_{Co}}(Y;X_1;X_2)=I(Y;X_1)-I(Y;X_1|X_2)=0$. 
\end{example}

It is also less clear if the co-information retains its intuitive appeal for higher-order interactions ($>2$ predictor variables), when the same state of a target RV $Y$ can have any combination of redundant, unique and (or) synergistic effects \cite{ref46}. 

The partial information (PI) decomposition framework (due to Williams and Beer \cite{ref2}) offers a solution to disentangle the redundant, unique and synergistic contributions to the total mutual information that a set of $K$ predictor RVs convey about a target RV. Consider the $K=2$ case. We use the following notation: $UI(\{{{X}_{1}}\};Y)$ and $UI(\{{{X}_{2}}\};Y)$ denote respectively, the unique information about $Y$ that $X_1$ and $X_2$ exclusively convey; ${{I}_{\cap}}(\{{{X}_{1}},{{X}_{2}}\};Y)$ is the redundant information about $Y$ that $X_1$ and $X_2$ both convey; $SI(\{{{X}_{1}}{{X}_{2}}\};Y)$ is the synergistic information about $Y$ that is conveyed only by the joint RV $(X_1,X_2)$. 

The governing equations for the PI decomposition are given in (15) \cite{ref2,ref3}.
\begin{align*}
  I({X_1}{X_2};Y) &= \underbrace {{I_ \cap }(\{ {X_1},{X_2}\} ;Y)}_{\text{redundant}} + \underbrace {SI(\{ {X_1}{X_2}\} ;Y)}_{\text{synergistic}}\\
                        &+ \underbrace {UI(\{ {X_1}\} ;Y) + UI(\{ {X_2}\} ;Y)}_{\text{unique}} \tag{15a}\\
  I({X_1};Y) &= {I_ \cap }(\{ {X_1},{X_2}\} ;Y) + UI(\{ {X_1}\} ;Y)  \tag{15b}\\
  I({X_2};Y) &= {I_ \cap }(\{ {X_1},{X_2}\} ;Y) + UI(\{ {X_2}\} ;Y)  \tag{15c}
\end{align*}
Using  the chain rule of MI, (15a)-(15c) implies
\begin{align*}
I({X_1};Y|{X_2}) = SI(\{ {X_1}{X_2}\} ;Y) + UI(\{ {X_1}\} ;Y)  \tag{15d}\\
I({X_2};Y|{X_1}) = SI(\{ {X_1}{X_2}\} ;Y) + UI(\{ {X_2}\} ;Y) \tag{15e}\\
I(Y;{X_1}) + UI(\{ {X_2}\} ;Y) = I(Y;{X_2}) + UI(\{ {X_1}\} ;Y) \tag{15f}
\end{align*}
From (15b)-(15e), one can easily see that the co-information is the difference between redundant and synergistic information. In particular, we have the following bounds.
\begin{align*}
-\min&\{I(X_1;Y|X_2),I(X_2;Y|X_1),I(X_1;X_2|Y)\}\\
&\leq {I_ \cap }(\{ {X_1},{X_2}\} ;Y)-SI(\{ {X_1}{X_2}\} ;Y)\\
&\leq \min\{I(X_1;Y),I(X_2;Y),I(X_1;X_2)\}\tag{15g}
\end{align*}
Equivalently, ${I_ \cap }(\{ {X_1},{X_2}\} ;Y) \leq SI(\{ {X_1}{X_2}\} ;Y)$ when there is any pairwise independence, i.e., when $X_1\perp X_2$, or $X_1\perp Y$, or $X_2\perp Y$, and ${I_ \cap }(\{ {X_1},{X_2}\} ;Y) \geq SI(\{ {X_1}{X_2}\} ;Y)$ when $(X_1,X_2,Y)$ form a Markov chain in any order, i.e., when $X_1-Y-X_2$, or $X_1-X_2-Y$ or $X_2-X_1-Y$. The following lemma gives conditions under which $I_\cap$ achieves its bounds.

\begin{lemma}
\leavevmode
    \begin{enumerate}
        \item[a)] If $X_1-X_2-Y$, then ${I_\cap}(\{ {X_1},{X_2}\};Y)=I({X_1};Y)$.
        \item[b)] If $X_2-X_1-Y$, then ${I_\cap}(\{ {X_1},{X_2}\};Y)=I({X_2};Y)$. 
        \item[c)] If $X_1-X_2-Y$ and $X_2-X_1-Y$, then ${I_\cap}(\{ {X_1},{X_2}\};$ $Y)=I({X_1};Y)=I({X_2};Y) = I(X_1 X_2;Y)$. 
        \item[d)] If $X_1-Y-X_2$, then ${I_\cap}(\{ {X_1},{X_2}\};Y)\geq I(X_1;X_2)$. 
    \end{enumerate}
\end{lemma}
\begin{proof}
The proofs follow directly from (15b)-(15e) and the symmetry of co-information.
\end{proof}

The following easy lemma gives the conditions under which the functions  $I_ \cap$, $UI$ and $SI$ vanish.

\begin{lemma}
\leavevmode
    \begin{enumerate}
        \item[a)] If $X_1\perp Y$ or $X_2\perp Y$, then ${I_ \cap }(\{ {X_1},{X_2}\};Y)=0$. Also, $X_1\perp X_2 \notimplies {I_ \cap }(\{ {X_1},{X_2}\};Y)=0$.
        \item[b)] If $X_1-X_2-Y$, then $UI(\{ {X_1}\} ;Y)=0$. Further, $SI(\{ {X_1}{X_2}\};Y)=0$, ${I_ \cap }(\{ {X_1},{X_2}\} ;Y)=I({X_1};Y)$, and $UI(\{ {X_2}\};Y)=I({X_2};Y|{X_1})$.
        \item[c)] If the predictor variables are identical or if either $X_1-X_2-Y$ or $X_2-X_1-Y$, then $SI(\{ {X_1}{X_2}\};Y)=0$. Also, if $\mathcal{Y}=\mathcal{X}_1\times \mathcal{X}_2$ and $Y=X_1X_2$, then $SI(\{ {X_1}{X_2}\};Y)=0$ and ${I_ \cap }(\{ {X_1},{X_2}\} ;Y)=I(X_1;X_2)$. 
    \end{enumerate}
\end{lemma}
\begin{proof}
The first part of a) is immediate from (15b) and (15c). The second part of a) is a direct consequence of the asymmetry built in the PI decomposition by distinguishing the predictor RVs ($X_1,X_2$) from the target RV ($Y$). Indeed, $X_1\perp X_2$ merely implies that ${I_ \cap }(\{ {X_1},{X_2}\};Y)=SI(\{ {X_1}{X_2}\};Y)-I(X_1;X_2|Y)$; the $\operatorname{RHS}$ does not vanish in general.
Part b) and c) follow directly from (15b)-(15e).
\end{proof}

We visualize the PI decomposition of the total mutual information $I({X_1}{X_2};Y)$ using a \emph{PI}-diagram \cite{ref2}. As detailed below, Fig. 1 shows the \emph{PI}-diagrams for the ``ideal'' PI decomposition of several canonical functions, viz., $\textsc{Copy}$ (and its degenerate simplifications $\textsc{Unq}$ and $\textsc{Rdn}$), $\textsc{Xor}$ and $\textsc{And}$ \cite{ref10,ref4}. Each irreducible PI atom in a \emph{PI}-diagram represents information that is either unique, synergistic or redundant. Ideally, one would like to further distinguish the redundancy induced by the function or mechanism itself (called \emph{functional} or \emph{mechanistic} redundancy) from that which is already present between the predictors themselves (called \emph{predictor} redundancy). However, at present it is not clear how these contributions can be disentangled, except for the special case of independent predictor RVs when the entire redundancy can be attributed solely to the mechanism \cite{ref10}.

\begin{example} 
Consider the $\textsc{Copy}$ function, $Y=\textsc{Copy}(X_1,X_2)$, where $Y$ consists of a perfect copy of $X_1$ and $X_2$, i.e., $Y=X_1X_2$ with $\mathcal{Y}=\mathcal{X}_1\times \mathcal{X}_2$. The $\textsc{Copy}$ function explicitly induces mechanistic redundancy and we expect that MI between the predictors completely captures this redundancy, i.e., ${I_ \cap }\left( {\{ {X_1},{X_2}\} ;({X_{{1}}},{X_{{2}}})} \right) = I({X_{{1}}};{X_{{2}}})$. Indeed, Lemma 5(c) codifies this intuition. 

(a) Fig. 1(a) shows the ideal PI decomposition for the distribution $p_{X_1X_2Y}$ with
$(00``00")=(01``01")=(11``11")=\tfrac{1}{3}$, where $(ab``ab")\defeq p_{{X_1}{X_2}{Y}}(a,b,ab)$. We then have ${I_ \cap }\left( {\{ {X_1},{X_2}\};({X_{{1}}},{X_{{2}}})} \right) = I({X_{{1}}};{X_{{2}}})=+.252$, $SI(\{ {X_1}{X_2}\};Y)=0$ and $UI(\{ {X_1}\};Y)=UI(\{ {X_2}\};Y)=+.667$.

(b) Fig. 1(b) shows the ideal PI decomposition for a simpler distribution $p_{X_1X_2Y}$ with $(00``00")=(01``01")=(10``10")=(11``11")=\tfrac{1}{4}$. Now $Y$ consists of a perfect copy of two i.i.d. RVs. Clearly, ${I_ \cap }\left( {\{ {X_1},{X_2}\} ;({X_{{1}}},{X_{{2}}})} \right)=0$. Since $SI(\{ {X_1}{X_2}\};Y)=0$ (vide Lemma 5(c)), only the unique contributions are nonzero, i.e., $UI(\{ {X_1}\};Y)=UI(\{ {X_2}\};Y)=+1$. We call this the $\textsc{Unq}$ function.

(c) Fig. 1(c) shows the ideal PI decomposition for the distribution $p_{X_1X_2Y}$ with $(000)=(111)=\tfrac{1}{2}$. This is an instance of a redundant $\textsc{Copy}$ mechanism with $X_1=X_2=Z$, where $Z=\operatorname{Bernoulli}(\tfrac{1}{2})$, so that $Y=X_1=X_2=Z$. We then have $I(X_1X_2;Y)={I_ \cap }\left( {\{ {X_1},{X_2}\} ;({X_{{1}}},{X_{{2}}})} \right)=1$. We call this the $\textsc{Rdn}$ function.
\end{example}

\begin{example}
Fig 1(d) captures the PI decomposition of the following distribution: $Y=\textsc{Xor}(X_1,X_2)$, where $X_i=\operatorname{Bernoulli}(\tfrac{1}{2})$, $i=1,2$. Only the joint RV $X_1X_2$ specifies information about $Y$, i.e., $I(X_1X_2;Y)=1$ whereas the singletons specify nothing, i.e., $I(X_i;Y)=0$, $i=1,2$. Neither the mechanism nor the predictors induce any redundancy since $I_\cap(\{ {X_1},{X_2}\} ;Y)=0$. $\textsc{Xor}$ is an instance of a purely synergistic function. 

Fig. 1(e) shows the ideal PI decomposition for the following distribution: $Y=(\textsc{Xor}(X_1',X_2'), (X_1'',X_2''),Z)$, where the predictor inputs are $X_1=(X_1',X_1'',Z)$ and $X_2=(X_2',X_2'',Z)$ with $X_1',X_2',X_1'',X_2'',Z$ i.i.d. The total MI of 4 bits is distributed equally between the four PI atoms. We call this the $\textsc{RdnUnqXor}$ function since it is a composition of the functions $\textsc{Rdn}$, $\textsc{Unq}$ and $\textsc{Xor}$. Also see Example 6 which gives an instance of composition of functions \textsc{Rdn} and \textsc{Xor}.
\end{example}

\begin{figure}[!t]
\centering
\includegraphics[width=3.5in]{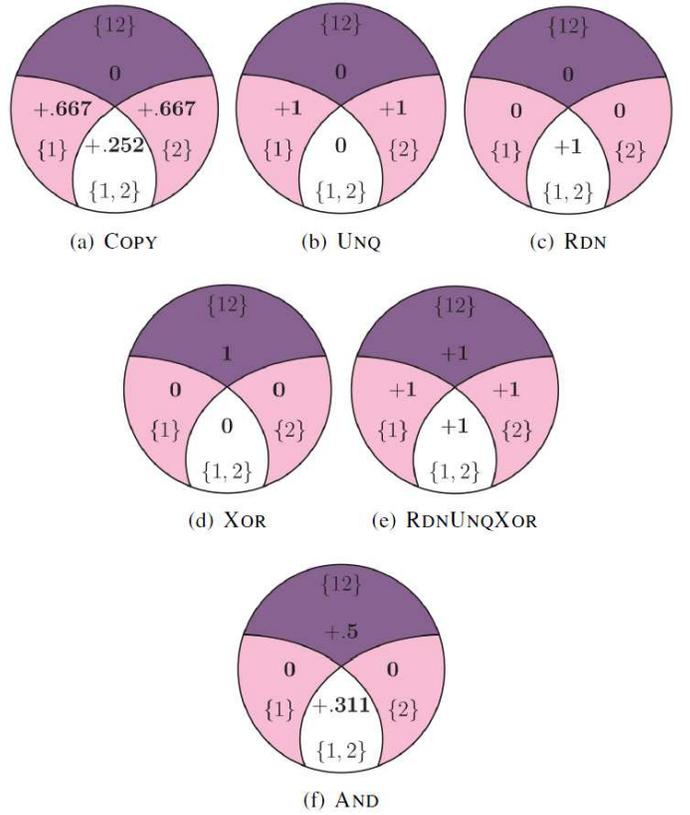}
\caption{\emph{PI}-diagrams showing the ``ideal'' PI decomposition of $I(X_1X_2;Y)$ for some canonical examples. $\{1\}$ and $\{2\}$ denote, resp. unique information about $Y$, that $X_1$ and $X_2$ exclusively convey; $\{1,2\}$ is the redundant information about $Y$ that $X_1$ and $X_2$ both convey; $\{12\}$ is the synergistic information about $Y$ that can only be conveyed by the joint RV $(X_1,X_2)$. (a) \textsc{Copy} (b) \textsc{Unq} (c) \textsc{Rdn} (d) \textsc{Xor} (e) \textsc{RdnUnqXor} (f) \textsc{And} (see description in text)} 
\label{fig:fig_pid}
\end{figure}

\begin{example}
Fig 1(f) shows the PI decomposition of the following distribution: $Y=\textsc{And}(X_1,X_2)$, where $X_i=\operatorname{Bernoulli}(\tfrac{1}{2})$, $i=1,2$ and $p_{X_1X_2Y}$ is such that $(000)=(010)=(100)=(111)=\tfrac{1}{4}$. The decomposition evinces both synergistic and redundant contributions to the total MI. The synergy can be explained as follows. First note that $X_1\perp X_2$, but $X_1\not\perp X_2|Y$ since $I(X_1;X_2|Y)=+.189\ne 0$. Fixing the output $Y$ induces correlations between the predictors $X_1$ and $X_2$ when there was none to start with. The induced correlations are the source of positive synergy.

Perhaps, more surprisingly, redundant information is not $0$ despite that $X_1\perp X_2$. The redundancy can be explained by noting that if either predictor input $X_1=0$ or $X_2=0$, then both $X_1$ and $X_2$ can exclude the possibility of $Y=1$. Hence the latter is nontrivial information shared between $X_1$ and $X_2$. This is clearer in light of the following argument that uses information structure aspects. Given the support of $p_{X_1X_2Y}$, the set of possible states of Nature include $\Omega=\{(000),(010),(100),(111)\}\defeq\{\omega_1,\omega_2,\omega_3,\omega_4\}$. $X_1$ generates the information partition ${\mathcal{P}_{X_1}} = {\omega _1}{\omega _2}|{\omega _3}{\omega _4}$. Likewise, $X_2$ generates the partition, ${\mathcal{P}_{X_2}} = {\omega _1}{\omega _3}|{\omega _2}{\omega _4}$. Let the true state of Nature be $\omega_1$. Consider the event $E=\{\omega_1,\omega_2,\omega_3\}$. Both $X_1$ and $X_2$ know $E$ at $\omega_1$, since ${\mathcal{P}_{X_1}}(\omega_1)=\{\omega_1,\omega_2\} \subset E$ and ${\mathcal{P}_{X_2}}(\omega_1)=\{\omega_1,\omega_3\} \subset E$. The event that $X_1$ knows $E$ is $K_{X_1}(E)=\{\omega_1,\omega_2\}$. Likewise, the event that $X_2$ knows $E$ is $K_{X_2}(E)=\{\omega_1,\omega_3\}$. Clearly, the event $K_{X_1}(E) \cap K_{X_2}(E)=\{\omega_1\}$ is known to both $X_1$ and $X_2$, so that $Y=1$ can be ruled out with probability of agreement one.

Indeed, for independent $X_1$ and $X_2$, when one can attribute the redundancy entirely to the mechanism, there is some consensus that $I_ \cap(\{ {X_1},{X_2}\} ;Y)=\tfrac{3}{4}\log\tfrac{4}{3}=+.311$ and $SI(\{ {X_1}{X_2}\};Y)=+.5$ \cite{ref2,ref5,ref10}.   
\end{example}

\begin{remark}
Independence of the predictor RVs implies a vanishing predictor redundancy but not necessarily a vanishing mechanistic redundancy (also see second part of Lemma 5(a)).
\end{remark}

As one final illustrative application of this framework, we consider the decomposition of Massey's directed information (DI) \cite{ref11} into PI atoms.
\begin{example}
For discrete-time stochastic processes ${{X}^{N}}$ and ${{Y}^{N}}$, the DI from $X$ to $Y$ is defined as follows. 
\[I({X^N} \to {Y^N}) \defeq \sum_{i = 1}^N I({X^i ;Y_i| Y^{i - 1}}),\]
where $X^i\defeq \{X_i,X_{i-1},\ldots\}$ denotes the past of $X$ relative to time $i$.
$I({{X}^{N}}\to {{Y}^{N}})$ answers the following operational question: Does consideration of the past of the process $X^N$ help in predicting the process $Y^N$ better than when considering the past of $Y^N$ alone? DI is a sum of conditional mutual information terms and admits an easy PI decomposition. 
\begin{align*}
  I(X^N \to Y^N) &=\sum\nolimits_{i = 1}^N {I({X^i};{Y_i}|{Y^{i - 1}})}\\
  &= \sum\nolimits_{i = 1}^N {UI(\{ {X^i}\} ;{Y_i}) + SI(\{ {X^i}{Y^{i - 1}}\} ;{Y_i})}, \tag{16}
\end{align*}
where we have used (15d) with $X_1\defeq X^i$, $X_2\defeq {{Y}^{i-1}}$ and $Y\defeq Y_i$. 

The decomposition has an intuitive appeal. Conditioning on the past gets rid of the common histories or redundancies shared between ${{X}^{i}}$ and ${{Y}^{i-1}}$ and adds in their synergy. Thus, given the knowledge of the past ${{Y}^{i-1}}$, information gained from learning ${{X}^{i}}$ has a unique component from ${{X}^{i}}$ alone as well as a synergistic component that comes from the interaction of ${{X}^{i}}$ and ${{Y}^{i-1}}$. The colored areas in Fig. 2 shows this decomposition of the ``local'' DI term $I({X^N} \to {Y^N})(i)$ into PI atoms, where $I({X^N} \to {Y^N}) = \sum\nolimits_{i = 1}^N I({X^N} \to {Y^N})(i)$.
\end{example}

\begin{figure}[!t]
\centering
\includegraphics[width=1.1in]{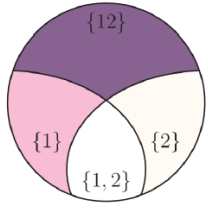}
\caption{\emph{PI}-diagram for the decomposition of Massey's directed information (DI). The colored areas correspond to the local DI term $I({X^N} \to {Y^N})(i)$, where $\{1\}=UI(\{{{X}^{i}}\};{{Y}_{i}})$, $\{12\}=SI(\{{{X}^{i}}{{Y}^{i-1}}\},{{Y}_{i}})$, $\{1,2\}=I_\cap(\{{{X}^{i}},{{Y}^{i-1}}\},{{Y}_{i}})$, and $\{2\}=UI(\{{{Y}^{i-1}}\};{{Y}_{i}})$ (see text)}
\label{fig:fig_DI}
\end{figure}

From (15a)-(15c), it is easy to see that the three equations specifying $I({{X}_{1}}{{X}_{2}};Y),\text{ }I({{X}_{1}};Y)$ and $I({{X}_{2}};Y)$ do not fully determine the four functions ${{I}_{\cap }}(\{{{X}_{1}},{{X}_{2}}\};Y)$, $UI(\{{{X}_{1}}\};Y)$, $UI(\{{{X}_{2}}\};Y)$ and $SI(\{{{X}_{1}}{{X}_{2}}\};Y)$. To specify a unique decomposition, one of the functions ${{I}_{\cap }}$, $SI$ or $UI$ needs to be defined or a fourth equation relating $I_\cap$, SI, and UI.

PI decomposition researchers have focused on axiomatically deriving measures of redundant \cite{ref2}, \cite{ref3}, \cite{ref7,ref8,ref9,ref10}, synergistic \cite{ref4} and unique information \cite{ref5}, \cite{ref6}. For instance, for a general $K$, any valid measure of redundancy ${{I}_{\cap }}({{X}_{1}},\ldots ,{{X}_{K}};Y)$ must satisfy the following \emph{basic} properties. Let ${{R}_{1}},\ldots ,{{R}_{k}}\subseteq \left\{ {{X}_{1}},\ldots ,{{X}_{K}} \right\}$, where $k\le K$.

\begin{itemize}[leftmargin=1.4cm]
    \item[\GP] Global Positivity: ${I_ \cap }(\{ {R_1}, \ldots ,{R_k}\} ;Y) \geq 0$.
    \item[\WS] Symmetry: ${I_ \cap }(\{ {R_1}, \ldots ,{R_k}\} ;Y)$ is invariant under reordering of the ${X_i}$'s.
    \item[\SR] Self-redundancy: ${I_ \cap }(R;Y) = I({X_R};Y)$. For instance, for a a single predictor $X_1$, the redundant information about the target $Y$ must equal $I(X_1;Y)$.
    \item[\Mo] Weak Monotonicity: ${I_ \cap }(\{ {R_1}, \ldots ,{R_{k - 1}},{R_k}\} ;Y) \leq {I_ \cap }(\{ {R_1},$ $\ldots ,{R_{k - 1}}\} ;Y)$ with equality if $\exists$ $R_i \in \{R_1,\ldots,R_k\}$ such that $H(R_iR_k)=H(R_k)$. 
    \item[\SMo] Strong Monotonicity: ${I_ \cap }(\{ {R_1}, \ldots ,{R_{k - 1}},{R_k}\} ;Y) \leq {I_ \cap }(\{ {R_1},$ $\ldots ,{R_{k - 1}}\} ;Y)$ with equality if $\exists$ $R_i \in \{R_1,\ldots,R_k\}$ such that $I(R_iR_k;Y)=I(R_k;Y)$. For the equality condition for $K=2$, also see Lemma 4(a)-(c).
    \item[\LN] Local Positivity: For all $K$, the derived PI measures are nonnegative. For instance for $K=2$, a nonnegative PI measure for synergy requires that $I(X_1X_2;Y)\geq {I_ \cup }(\{X_1,X_2\};Y)$, where $I_\cup$ is the \emph{union information} which is related to $I_\cap$ (for any $K$) by the inclusion-exclusion principle \cite{ref2}.
    \item[\Idty] Identity: For $K=2$, ${I_ \cap }\left( {\{ {X_1},{X_2}\} ;({X_{1}},{X_{2}})} \right) = I({X_{{1}}};{X_{{2}}})$ \cite{ref10}. 
\end{itemize}
    
The following properties capture the behavior of an ideal ${I_ \cap }$ when one of the predictor or target arguments is enlarged.
\begin{itemize}[leftmargin=1.4cm]
    \item[\TM] Target Monotonicity: If $H(Y|Z)=0$, then ${I_ \cap }(\{ {R_1}, \ldots ,{R_k}\} ;Y)\leq {I_ \cap }(\{ {R_1},$ $ \ldots ,{R_k}\};Z)$.
    \item[\PM] Predictor Monotonicity: If $H(R_1|R_1')=0$, then ${I_ \cap }(\{ {R_1}, \ldots ,{R_k}\} ;Y)\leq {I_ \cap }(\{ {R_1'},R_2,\ldots ,{R_k}\};Y)$.
\end{itemize}

A similar set of monotonicity properties are desirable of an ideal $UI$. We consider only the $K=2$ case and write $UI_{X_2}(\{{{X}_{1}}\};Y)$ to explicitly specify the information about $Y$ exclusively conveyed by $X_1$.
\begin{itemize}[leftmargin=1.4cm]
    \item[\TMU] Target Monotonicity: If $H(Y|Z)=0$, then $UI_{X_2}(\{{{X}_{1}}\};Y) \leq UI_{X_2}(\{{{X}_{1}}\};Z)$.
    \item[\PMU] Predictor Monotonicity: If $H(X_1|X_1^\prime)=0$, then $UI_{X_2}(\{{{X}_{1}}\};Y) \leq UI_{X_2}(\{X_1^\prime \};Y)$.
    \item[\PMUC] Predictor Monotonicity with respect to the complement: If $H(X_2|X_2^\prime)=0$, then $UI_{X_2^\prime}(\{{{X}_{1}}\};Y) \leq UI_{X_2}(\{X_1\};Y)$.
\end{itemize}

Properties \textbf{(M)} and \textbf{(SM)} ensure that any reasonable measure of redundancy is monotonically nonincreasing with the number of predictors. For a general $K$, given a measure of redundant information that satisfies \textbf{(S)} and \textbf{(M)}, only those subsets need to be considered which satisfy the ordering relation ${{R}_i}\nsubseteq{{R}_j},\forall i\ne j$ (i.e., the family of sets ${{R}_{1}},\ldots ,{{R}_{k}}$ forms an antichain) \cite{ref2}, \cite{ref3}. Define a partial order $\precsim$ on the set of antichains by the relation: ($S_1,\ldots,S_m$) $\precsim$ (${{R}_{1}},\ldots ,{{R}_{k}}$) iff for each $j=1,\ldots,k$ $\exists$ $i\leq m$ such that $S_i \subseteq R_j$. Then, equipped with $\precsim$, the set of antichains form a lattice ${\mathfrak{L}}$ called the PI or the redundancy lattice \cite{ref2}. By virtue of \textbf{(M)}, for a fixed $Y$, ${I_ \cap }(\{ {R_1}, \ldots ,{R_k}\} ;Y)$ is a monotone function with respect to $\precsim$. Then, a unique decomposition of the total mutual information is accomplished by associating with each element of ${\mathfrak{L}}$ a PI measure ${I_\partial }$ which is the M\"{o}bius transform of ${I_ \cap }$ so that we have ${I_ \cap }(\{ {R_1}, \ldots ,{R_k}\} ;Y)=\sum\limits_{({S_1}, \ldots ,{S_m})\precsim({R_1}, \ldots ,{R_k})} {{I_\partial }({S_1},\ldots,{S_m};Y)}$. 

For instance, for $K=2$ (see (15a)), the PI measures are ${{I}_{\partial}}(\{{{X}_{1}},{{X}_{2}}\};Y)={{I}_{\cap }}(\{{{X}_{1}},{{X}_{2}}\};Y)$, ${{I}_{\partial}}(\{{{X}_{1}}\};Y)=UI(\{{{X}_{1}}\};Y)$, ${{I}_{\partial}}(\{{{X}_{2}}\};Y)=UI(\{{{X}_{2}}\};Y)$, and ${{I}_{\partial}}(\{{{X}_{1}}{{X}_{2}}\};Y)=SI(\{{{X}_{1}}{{X}_{2}}\};Y)$.

While elegant in its formulation, the lattice construction does not by itself guarantee a \emph{nonnegative} decomposition of the total mutual information. The latter depends on the chosen measure of redundancy used to generate the PI decomposition. Given a measure of redundant information, some of the recurrent pathologies reported thus far include incompatibility of properties (a) \textbf{(LP)} and \textbf{(TM)} \cite{ref4}, \cite{ref8}, \cite{ref10}, \cite{ref5}, \cite{ref6}, (b) \textbf{(LP)} and \textbf{(Id)} for $K\geq 3$ \cite{ref6}, and (c) \textbf{(TM)} and \textbf{(Id)} for $K=2$, whenever there is mechanistic dependence between the target and the predictors \cite{ref6}. For a nonvanishing mechanistic dependency, \textbf{(TM)} and \textbf{(Id)} are incompatible since together they imply ${I_ \cap }\left({\{{X_1},{X_2}\};({X_1},{X_2})}\right) \leq I(X_1;X_2)$. For example, the desired decomposition of \textsc{And} in Example 9 contradicts \textbf{(TM)}. None of the measures of $I_\cap$ proposed thus far satisfies \textbf{(TM)}. 
In the next section, we restrict ourselves to the bivariate case as some of the pathological features are already manifest.

\subsection{Measures of Redundant Information Based on Common Information}

In this section, we dwell on the relationship between redundant information and the more familiar information-theoretic notions of common information. In particular, we seek to answer the following question: can optimization over a single RV yield a plausible measure of redundancy that satisfies \LN?

A simple measure of redundant information between predictors $({{X}_{1}},{{X}_{2}})$ about a target RV $Y$ is defined as follows \cite{ref9}.
\begin{align*}
I_ \cap ^1(\{ {X_1},{X_2}\} ;Y) &= \mathop {\max }\limits_{Q:\text{ }H(Q|{X_1}) = H(Q|{X_2}) = 0} I(Q;Y)\\
&= I({X_1} \wedge {X_2};Y)\tag{17}
\end{align*}
$I_{\cap }^{1}$ satisfies \textbf{(GP)}, \textbf{(S)}, \textbf{(I)}, \textbf{(M)} and \textbf{(TM)} but not \textbf{(Id)} \cite{ref9}. 
$I_{\cap }^{1}$ inherits the negative character of the original definition of GK and fails to capture any redundancy beyond a certain deterministic interdependence between the predictors. Unless ${{p}_{{{X}_{1}}{{X}_{2}}}}$ is decomposable, $I_ \cap ^1(\{ {X_1},{X_2}\} ;Y)$ is trivially zero, even if it is the case that the predictors share nontrivial redundant information about the target $Y$. Furthermore, $I_{\cap }^{1}$ violates \LN \cite{ref9} and is too restrictive in the sense that it does not capture the full informational overlap. 

One can relax the constraint in (17) in a natural way by using the asymmetric notion of CI, $C^2(\{ {X_1},{X_2}\} ;Y)$ introduced earlier in (12). For consistency of naming convention, we call this $I_\cap^2$. 
\begin{align*}
I_\cap^2(\{ {X_1},{X_2}\} ;Y) = \mathop {\max }\limits_{\substack{ Q:\text{ }Q-X_1-Y\\ \hspace{4mm} Q-X_2-Y}} I(Q;Y) \tag{18}
\end{align*}
The definition has an intuitive appeal. If $Q$ specifies the optimal redundant RV, then conditioning on any predictor $X_i$ should remove all the redundant information about $Y$, i.e., $I(Q;Y|{{X}_{i}})=0$, $i=1,2$ \cite{ref8}. $I_{\cap }^{2}$ remedies the degenerate nature of $I_{\cap }^{1}$ with respect to indecomposable distributions \cite{ref8}. It is also easy to see that the derived unique information measure, $UI^2$ is nonnegative. 
\begin{align*}
UI^2(\{ {X_1}\} ;Y) = \mathop {\min }\limits_{\substack{ Q:\text{ }Q-X_1-Y\\ \hspace{4mm} Q-X_2-Y}}  I({X_1};Y|Q) \tag{19}
\end{align*}
$UI^2$ readily satisfies the symmetry condition (15f) since given $Q$ such that $Q - {X_1} - Y$ and $Q - {X_2} - Y$, we have ${I({X_1};Y) + I({X_2};Y|Q)}\mathop= \limits^{{\text{(a)}}}{I(Q{X_1};Y) + I({X_2};Y|Q)}={I(Q{X_2};Y) + I({X_1};Y|Q)}\mathop= \limits^{{\text{(b)}}}{I({X_2};Y) + I({X_1};Y|Q)}$, where (a) follows from $Q-{{X}_{1}}-Y$ and (b) follows from $Q-{{X}_{2}}-Y$. 

For the proofs of Lemma 6 and 7 to follow, we shall use the standard facility of Information diagrams ($I$-diagrams) \cite{ref34a}. 
For finite RVs, there is a one-to-one correspondence between Shannon's information measures and a signed measure $\mu^*$ over sets, called the $I$-measure. We denote the $I$-Measure of RVs $(Q,X_1,X_2,Y)$ by $\mu^*$. For a RV $X$, we overload notation by using $X$ to also label the corresponding set in the $I$-diagram. Note that the $I$-diagrams in Fig. 3 are valid information diagrams since the sets $Q,X_1,X_2,Y$ intersect each other generically and the region representing the set $Q$ splits each atom into two smaller ones.

\begin{lemma}
If $X_1\perp X_2$, then $I_\cap^2(\{ {X_1},{X_2}\} ;Y)=0$.
\end{lemma}

\begin{proof}
The atoms on which $\mu^*$ vanishes when the Markov chains $Q-X_1-Y$ and $Q-X_2-Y$ hold and $X_1\perp X_2$ are shown in the generic $I$-diagram in Fig. 3(a); $\mu^*(Q \cap Y)=0$ which gives the result. 
\end{proof}

\begin{lemma}
If $X_1-Y-X_2$, then $I_\cap^2(\{ {X_1},{X_2}\};Y)\leq I(X_1;X_2)$.
\end{lemma}

\begin{proof}
The atoms on which $\mu^*$ vanishes when the Markov chains $Q-X_1-Y$, $Q-X_2-Y$ and $X_1-Y-X_2$ hold are shown in the $I$-diagram in Fig. 3(b). In general, for the atom $X_1 \cap X_2 \cap Y$, $\mu^*$ can be negative. However, since $X_1-Y-X_2$ is a Markov chain by assumption, we have $\mu^*(X_1 \cap X_2 \cap Y)=\mu^*(X_1 \cap X_2) \geq 0$. Then $\mu^*(Q \cap Y) \leq \mu^*(X_1 \cap X_2)$, which gives the desired claim.
\end{proof}

By Lemma 7, $I_\cap^2$ already violates the requirement posited in Lemma 4(d) for an ideal $I_\cap$. It turns out that we can make a more precise statement under a stricter assumption, which also amounts to proving that $I_\cap^2$ violates \textbf{(Id)}.  

\begin{figure}[!t]
\centering
\includegraphics[width=3.2in]{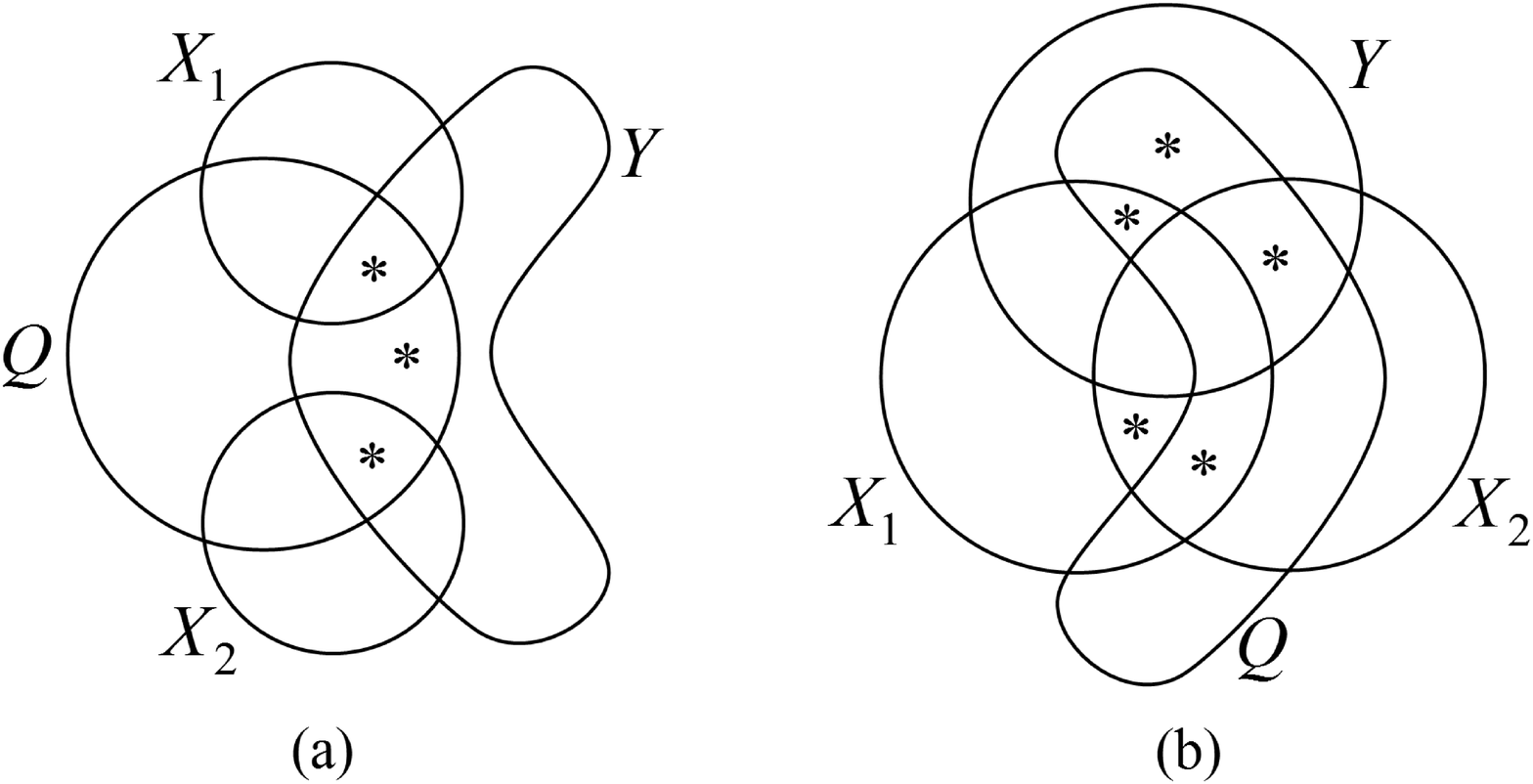}
\caption{\emph{I}-diagrams for proofs of (a) Lemma 6 and (b) Lemma 7. Denoting the $I$-Measure of RVs $(Q,X_1,X_2,Y)$ by $\mu^*$, the atoms on which $\mu^*$ vanishes are marked by an asterisk (see text)}\vspace{-1em}
\label{fig_Idiag_proofs}
\end{figure}

\begin{lemma}
Let $\mathcal{Y}=\mathcal{X}_1\times \mathcal{X}_2$ and $Y=X_1X_2$. Then $I_\cap^2(\{ {X_1},{X_2}\} ;Y)=C_{GK}(X_1;X_2)\leq I(X_1;X_2)$.
\end{lemma}

\begin{proof}
First note that 
\begin{align*}
I_ \cap^2(\{ {X_1},{X_2}\} ;X_1X_2)&=\mathop {\max }\limits_{\substack{ Q:\text{ }Q-X_1-X_1X_2\\ \hspace{4mm} Q-X_2-X_1X_2}} I(Q;X_1X_2)\\
&=\mathop {\max }\limits_{\substack{ Q:\text{ }Q-X_1-X_2\\ \hspace{4mm} Q-X_2-X_1}} I(Q;X_1).
\end{align*}
From Lemma 2 we have that, given $p_{Q|X_1X_2}$ such that $Q-X_1-X_2$ and $Q-X_2-X_1$, $\exists\text{ }p_{Q'|X_1X_2}$ such that $H(Q'|X_1)=H(Q'|X_2)=0$ and $X_1X_2-Q'-Q$. Then $I(Q;X_1)=I(Q;X_1X_2)=I(Q;X_1X_2Q')=I(Q';Q)\leq H(Q')$. $Q'$ is the maximal common RV of $X_1$ and $X_2$. Thus, we have $\mathop {\max }\limits_{\substack{ Q:\text{ }Q-X_1-X_2\\ \hspace{4mm} Q-X_2-X_1}} I(Q;X_1)=H(Q')\leq I(X_1;X_2)$, with equality iff $X_1-Q-X_2$, or equivalently, iff $(X_1,X_2)$ is saturable (see Remark 3). 
\end{proof}

\begin{remark}
Consider the G{\'a}cs-K\"{o}rner version of $I_\cap^2$: 
$$C_\cap^2(\{ {X_1},{X_2}\} ;Y) = \mathop {\max }\limits_{\substack{ Q:\text{ }Q-X_1-Y\\ \hspace{4mm} Q-X_2-Y}} C_{GK}(Q;Y).$$
Interestingly, $C_\cap^2$ satisfies \textbf{(Id)} in the sense that $C_\cap^2(\{ {X_1},{X_2}\} ;X_1X_2)=C_{GK}(X_1;X_2)$, or equivalently, $\mathop {\max }\limits_{\substack{ Q:\text{ }Q-X_1-X_2\\ \hspace{4mm} Q-X_2-X_1}} C_{GK}(Q;X_1X_2)=C_{GK}(X_1;X_2)$. To show this, we again use Lemma 2. Clearly the $Q$ that achieves the maximum is the maximal common RV $Q'$ so that we have, $\operatorname{LHS}=C_{GK}(Q';X_1X_2)=C_{GK}(X_1\wedge X_2;X_1X_2)=H(X_1\wedge X_2\wedge X_1X_2)=H(X_1\wedge X_2)=C_{GK}(X_1;X_2)=\operatorname{RHS}$.
\end{remark}

\begin{proposition}
$I_{\cap }^{2}$ satisfies \GP$\text{, }$ \WS$\text{, }$ \SR$\text{, }$ \Mo$\text{, }$ and \SMo$\text{ }$ but not \LN$\text{ }$ and \Idty.
\end{proposition}

\begin{proof}

\GP$\text{ }$ Global positivity follows immediately from the nonnegativity of mutual information.

\vspace{2mm}
\WS$\text{ }$ Symmetry follows since $I_{\cap }^{2}$ is invariant under reordering of the ${X_i}$'s.

\vspace{2mm}
\SR$\text{ }$ If $Q-X_1-Y$, then $I(Q;Y)\leq I(X_1;Y)$. Then, self-redundancy follows from noting that $I_{\cap }^{2}(\{X_1\};Y)=\mathop {\max }\limits_{Q:\text{ }Q - {X_1} - Y} I(Q;Y)=I(X_1;Y)$. 

\vspace{2mm}
\Mo$\text{ }$ We first show that $I_{\cap }^{2}(\{X_1,X_2\};Y)\leq I_{\cap }^{2}(\{X_1\};Y)$. This follows immediately from noting that $\mathop {\max }\limits_{\substack{ Q:\text{ }Q-X_1-Y\\ \hspace{4mm} Q-X_2-Y}} I(Q;Y)$ $\leq \mathop {\max }\limits_{Q:\text{ }Q - {X_1} - Y} I(Q;Y)$, since the constraint set for the $\operatorname{LHS}$ is a subset of that for the $\operatorname{RHS}$ and the objective function for the maximization is the same on both sides.

For the equality condition, we need to show that if $H(X_1|X_2)=0$, then $I_{\cap }^{2}(\{X_1,X_2\};Y)=I_{\cap }^{2}(\{X_1\};Y)$. It suffices to show that if $H(X_1|X_2)=0$, then $\mathop {\max }\limits_{\substack{ Q:\text{ }Q-X_1-Y\\ \hspace{4mm} Q-X_2-Y}} I(Q;Y) \geq \mathop {\max }\limits_{Q:\text{ }Q - {X_1} - Y} I(Q;Y)$. This holds since if $H(X_1|X_2)=0$, then $Q-X_1-Y\implies Q-X_2-Y$. 

\vspace{2mm}
\SMo$\text{ }$ Since \textbf{(M)} holds, it suffices to show the equality condition. For the latter, we need to show that if $I(X_1X_2;Y)=I(X_2;Y)$ or equivalently, if $X_1-X_2-Y$, then $I_{\cap }^{2}(\{X_1,X_2\};Y)=I_{\cap }^{2}(\{X_1\};Y)$. This follows from noting that $I(Q;Y)\mathop\leq\limits^{{\text{(a)}}}I(X_1;Y)\mathop\leq\limits^{{\text{(b)}}}I(X_2;Y)$, where (a) follows from $Q-X_1-Y$ and (b) follows from $X_1-X_2-Y$. Hence, we have $I_{\cap }^{2}(\{X_1,X_2\};Y)=\mathop {\max }\limits_{\substack{ Q:\text{ }Q-X_1-Y\\ \hspace{4mm} Q-X_2-Y}} I(Q;Y)=I(X_1;Y)\mathop= \limits^{{\text{(c)}}}I_{\cap }^{2}(\{X_1\};Y)$, where (c) follows from \textbf{(I)}.

\vspace{2mm}
\LN$\text{ }$ Proof by counter-example: We show that if $X_1-Y-X_2$, then \textbf{(LP)} is violated. First note that if $X_1-Y-X_2$, then using the symmetry of co-information, the derived synergy measure is $SI^2(\{ {X_1}{X_2}\} ;Y)=I_{\cap }^{2}(\{X_1,X_2\};Y)-I(X_1;X_2)$. From Lemma 7, it follows that $I_{\cap }^{2}(\{X_1,X_2\};Y)\leq I(X_1;X_2)$ so that $SI^2(\{ {X_1}{X_2}\} ;Y)\leq 0$. Hence, there exists at least one distribution such that \textbf{(LP)} does not hold, which suffices to say that \textbf{(LP)} does not hold in general. 

Indeed, the \textsc{Copy} function in Example 7 provides a direct counterexample, since $I_{\cap }^{2}(\{X_1,X_2\};Y)=0$ and $SI^2(\{ {X_1}{X_2}\} ;Y)=-I(X_1;X_2)\leq 0$. Not surprisingly, the derived synergy measure exactly matches the deficit in mechanistic redundancy that $I_ \cap ^2$ fails to capture. 

\vspace{2mm}
\Idty$\text{ }$ By Lemma 8, $I_\cap^2$ violates \textbf{(Id)}. 

\end{proof}

\begin{proposition}
$I_{\cap }^{2}$ satisfies \PM$\text{ }$ but not \TM.
\end{proposition}

\begin{proof}

\TM$\text{ }$ We need to show that if $H(Y|Z)=0$, then, 
$I_ \cap^2(\{ {X_1},{X_2}\} ;Y) \leq I_ \cap^2(\{ {X_1},{X_2}\} ;Z)$, or equivalently, $\mathop {\max }\limits_{\substack{ Q:\text{ }Q-X_1-Y\\ \hspace{4mm} Q-X_2-Y}} I(Q;Y) \leq \mathop {\max }\limits_{\substack{ Q:\text{ }Q-X_1-Z\\ \hspace{4mm} Q-X_2-Z}} I(Q;Z)$. The latter does not hold since $Q-X_i-Z\implies Q-X_i-Y$, $i=1,2$, but the converse does not hold in general. Hence, $I_{\cap }^{2}$ violates \textbf{(TM)}.

\vspace{2mm}
\PM$\text{ }$ We need to show that if $H(X_1|X_1')=0$, then $I_ \cap^2(\{ {X_1},{X_2}\} ;Y) \leq I_ \cap^2(\{ {X_1'},{X_2}\} ;Y)$, or equivalently, $\mathop {\max }\limits_{\substack{ Q:\text{ }Q-X_1-Y\\ \hspace{4mm} Q-X_2-Y}} I(Q;Y) \leq \mathop {\max }\limits_{\substack{ Q:\text{ }Q-X_1'-Y\\ \hspace{4mm} Q-X_2-Y}} I(Q;Y)$. The latter holds since $Q-X_1-Y\implies Q-X_1'-Y$. Since $I_{\cap }^{2}$ is symmetrical in the $X_i$'s, $I_{\cap }^{2}$ satisfies \textbf{(PM)}. 
\end{proof}

\begin{proposition}
$UI^2$ satisfies \TMU$\text{ }$ and \PMUC$\text{ }$ but not \PMU.
\end{proposition}

\begin{proof}

\TMU$\text{ }$ We need to show that if $H(Y|Z)=0$, then $UI^2_{X_2}(\{{{X}_{1}}\};Y) \leq UI^2_{X_2}(\{{{X}_{1}}\};Z)$, or equivalently, $\mathop {\min }\limits_{\substack{ Q:\text{ }Q-X_1-Y\\ \hspace{4mm} Q-X_2-Y}}  I({X_1};Y|Q) \leq \mathop {\min }\limits_{\substack{ Q:\text{ }Q-X_1-Z\\ \hspace{4mm} Q-X_2-Z}}  I({X_1};Z|Q)$. The latter holds since $Q-X_i-Z\implies Q-X_i-Y$, $i=1,2$ and $I({X_1};Y|Q) \leq I({X_1};Z|Q)$. 

\vspace{2mm}
\PMUC$\text{ }$ We need to show that if $H(X_2|X_2')=0$, then $UI^2_{X_2}(\{{{X}_{1}}\};Y) \geq UI^2_{X_2'}(\{X_1\};Y)$, or equivalently, $\mathop {\min }\limits_{\substack{ Q:\text{ }Q-X_1-Y\\ \hspace{4mm} Q-X_2-Y}}  I({X_1};Y|Q) \geq \mathop {\min }\limits_{\substack{ Q:\text{ }Q-X_1-Y\\ \hspace{4mm} Q-X_2'-Y}}  I({X_1};Y|Q)$. The latter holds since $Q-X_2-Y\implies Q-X_2'-Y$.

\vspace{2mm}
\PMU$\text{ }$ We need to show that if $H(X_1|X_1')=0$, then $UI^2_{X_2}(\{{{X}_{1}}\};Y) \leq UI^2_{X_2}(\{X_1'\};Y)$, or equivalently, $\mathop {\min }\limits_{\substack{ Q:\text{ }Q-X_1-Y\\ \hspace{4mm} Q-X_2-Y}}  I({X_1};Y|Q) \leq \mathop {\min }\limits_{\substack{ Q:\text{ }Q-X_1'-Y\\ \hspace{4mm} Q-X_2-Y}}  I({X_1'};Y|Q)$. The latter does not hold since $Q-X_1-Y\implies Q-X_1'-Y$, but the converse does not hold in general.
\end{proof}

\subsection{Comparison with Existing Measures}

For the $K=2$ case, it is sufficient to specify any one of the functions $I_{\cap}$, $UI$ or $SI$ to determine a unique decomposition of $I(X_1X_2;Y)$ (see (15a)). Information-geometric arguments have been forwarded in \cite{ref10}, \cite{ref3} to quantify redundancy. We do not repeat all the definitions in \cite{ref10}. However, for the sake of exposition, we prefer working with the unique information since geometrically, the latter shares some similarities with the mutual information which can be interpreted as a weighted distance.
\begin{align*}
I({X_1};Y) = \sum\nolimits_{x \in {{\mathcal{X}}_1}} {{p_{{X_1}}}(x)} D({p_{Y|{X_1} = {x_1}}}||{p_Y}).
\end{align*}

Given a measurement of a predictor, say $X_1=x_1$, unique information is defined in terms of the reverse information projection \cite{ref40} of  ${{p}_{Y}}_{|{{X}_{1}}={{x}_{1}}}$ on the convex closure of the set of all conditional distributions of $Y$ for all possible outcomes of ${{X}_{2}}$.
\begin{align*}
UI^{3}(\{ {X_1}\} ;Y) = \sum\nolimits_{x \in {{\mathcal{X}}_1}} {{p_{{X_1}}}(x)} \mathop {\min }\limits_{Q \in \Delta } D({p_{Y|{X_1} = {x_1}}}||Q),\tag{20}
\end{align*}
where $\Delta $ is the convex hull of ${{\left\{ {{p}_{Y}}_{|{{X}_{2}}={{x}_{2}}} \right\}}_{{{x}_{2}}\in {{\mathsf{\mathcal{X}}}_{2}}}}$, the family of all conditional distributions of $Y$ given the different outcomes of $X_2$. Since the KL divergence is convex with respect to both its arguments, the minimization in (20) is well-defined. It is easy to see however, that $UI^{3}$ violates the symmetry condition (15f) unless the projection is guaranteed to be unique. Uniqueness is guaranteed only when the set we are projecting onto is log-convex \cite{ref40}. In particular, (20) only gives a lower bound on the unique information so that we have,
\begin{align*}
UI^{3}(\{ {X_1}\} ;Y) \geq \sum\nolimits_{x \in {{\mathcal{X}}_1}} {{p_{{X_1}}}(x)} \mathop {\min }\limits_{Q \in \Delta } D({p_{Y|{X_1} = {x_1}}}||Q).
\end{align*}
The symmetry is restored by considering the minimum of the projected information terms for the derived redundant information \cite{ref10}.
\begin{align*}
I_\cap^{3}(\{ {X_1},{X_2}\} ;Y)=\min [&(I(X_1;Y)-UI^{3}(\{ {X_1}\} ;Y)),\\&(I(X_2;Y)-UI^{3}(\{ {X_2}\} ;Y))]. \tag{21}
\end{align*}
$I_\cap^{3}$ satisfies \textbf{(GP)}, \textbf{(S)}, \textbf{(I)}, \textbf{(M)}, \textbf{(LP)} and \textbf{(Id)} but not \textbf{(TM)} \cite{ref10}.  

The following measure of unique information is proposed in \cite{ref5}.
\begin{align*}
UI_\cap^{4}(\{ {X_1}\} ;Y)=\mathop {\max }\limits_{\substack{(X_1',X_2',Y'):\text{ }p_{X_1'Y'}=p_{X_1Y},\\\hspace{15.5mm}p_{X_2'Y'}=p_{X_2Y}}}  I({X_1'};Y'|X_2'). \tag{22}
\end{align*}
The derived redundant information is $I_\cap^{4}(\{ {X_1},{X_2}\} ;Y)=I(X_1;Y)-UI_\cap^{4}(\{ {X_1}\} ;Y)$.
$I_\cap^{4}$ satisfies \textbf{(GP)}, \textbf{(S)}, \textbf{(I)}, \textbf{(M)}, \textbf{(LP)} and \textbf{(Id)} but not \textbf{(TM)} \cite{ref5}.

Proposition 4 shows that both $I_\cap^{3}$ and $I_\cap^{4}$ satisfy \textbf{(SM)}.

\begin{proposition}
$I_\cap^{3}$ and $I_\cap^{4}$ satisfy \SMo$\text{. }$
\end{proposition}

\begin{proof}
See Lemma 13 and Corollary 23 in \cite{ref5}.
\end{proof}
\smallskip

It is easy to show that $I_\cap^{1}$ violates \textbf{(SM)} (see Example \textsc{ImperfectRdn} in \cite{ref8}).
Table 1 lists the desired properties satisfied by $I_\cap^{1}$, $I_\cap^{2}$, $I_\cap^{3}$ \cite{ref10} and $I_\cap^{4}$ \cite{ref5}.

The following proposition from \cite{ref5} gives the conditions under which $I_\cap^{3}$ and $I_\cap^{4}$ vanish.
\begin{proposition}
If both $X_1-Y-X_2$ and $X_1\perp X_2$ hold, then $I_\cap^{3}=I_\cap^{4}=0$.
\end{proposition}

\begin{proof}
See Corollary 10 and Lemma 21 in \cite{ref5}.
\end{proof}
\smallskip

In general, the conditions for which an ideal $I_\cap$ vanishes are given in Lemma 5(a). Indeed, if both $X_1-Y-X_2$ and $X_1\perp X_2$ hold, then from (15g) we have that ${I_ \cap }(\{ {X_1},{X_2}\} ;Y)-SI(\{ {X_1}{X_2}\} ;Y)=0$, so that $I_\cap\geq 0$ in general (also see Lemma 4(d)). However, we have not been able to produce a counterexample to refute Proposition 5 (for an ideal $I_\cap$). We conjecture that the conditions $X_1-Y-X_2$ and $X_1\perp X_2$ are ideally sufficient for a vanishing $I_\cap$.

Proposition 5 highlights a key difference between $I_\cap^{2}$ and the related measures $I_\cap^{3}$ and $I_\cap^{4}$. By Lemma 6, we have that $I_\cap^{2}$ vanishes if $X_1\perp X_2$. Clearly, unlike $I_\cap^{3}$ and $I_\cap^{4}$, $I_\cap^{2}$ is not sensitive to the extra Markov condition $X_1-Y-X_2$. This is most clearly evident for the \textsc{And} function in Example 9, where we have $I(X_1;X_2)=0$ and $I(X_1;X_2|Y)=+.189$. Lemma 6 dictates that $I_\cap^{2}=0$ if $X_1\perp X_2$ and the ensuing decomposition is degenerate. Thus, for independent predictor RVs, if $Y$ is a function of $X_1$ and $X_2$ when any positive redundancy can be attributed solely to the \emph{common effect} $Y$, $I_\cap^{2}$ fails to capture the required decomposition (see Remark 5). When the predictor RVs are not independent, a related degeneracy is associated with the violation of \textbf{(Id)} when $I_\cap^{2}$ fails to attain the mutual information between the predictor RVs (see the \textsc{Copy} function in Example 7). Indeed, by Lemma 8, ${I_ \cap }(\{ {X_1},{X_2}\} ;Y)=I(X_1;X_2)$ iff $p_{X_1X_2}$ is saturable. Also by Lemma 7, $I_\cap^{2}$ violates the requirement posited in Lemma 4(d) which generalizes the \textbf{(Id)} property.
Interestingly, Lemma 6 also shows that any reasonable measure of redundant information cannot be derived by optimization over a single RV.

\begin{table}[t]
\renewcommand{\arraystretch}{1.3}
\caption{Desired properties of $I_\cap$ satisfied by the CI-based measures $I_\cap^{1}$ and $I_\cap^{2}$, and the earlier measures $I_\cap^{3}$ \cite{ref10} and $I_\cap^{4}$ \cite{ref5}}
\label{Properties}
\centering
\begin{tabular}{ l l l l l l }
\toprule
\multicolumn{2}{c}{\textbf{Property}} & \textbf{$I_\cap^{1}$} & \textbf{$I_\cap^{2}$} & \textbf{$I_\cap^{3}$} & \textbf{$I_\cap^{4}$} \\
\midrule
 \GP & Global Positivity  & \checkmark & \checkmark & \checkmark & \checkmark \\
\addlinespace
 \WS & Weak Symmetry & \checkmark & \checkmark & \checkmark & \checkmark \\
\addlinespace
 \SR & Self-redundancy & \checkmark & \checkmark & \checkmark & \checkmark \\
\addlinespace
 \Mo & Weak Monotonicity & \checkmark & \checkmark & \checkmark & \checkmark \\
\addlinespace
 \SMo & Strong Monotonicity &  & \checkmark & \checkmark & \checkmark \\
\addlinespace
 \LN & Local Positivity &  &  & \checkmark & \checkmark \\
 \addlinespace
 \Idty & Identity &  &  & \checkmark & \checkmark \\
\bottomrule
\end{tabular}
\end{table}

We give one final example which elucidates the subtlety of the PI decomposition problem from a coding-theoretic point of view. 
Consider the distribution in Example 11, where $Y=\textsc{Copy}(X_1,X_2)$. 
The PI decomposition of $I(X_1X_2;Y)$ in this case reduces to the decomposition of $H(X_1X_2)$ into redundant and unique information contributions\footnote{Also see Example 5.}. 
\begin{example}
Let ${\mathcal{X}}_1=\{1,2\}$ and ${\mathcal{X}}_2=\{3,4,5,6\}$. 
Let $Y=X_1X_2$ with $\mathcal{Y}={\mathcal{X}}_1\times {\mathcal{X}}_2$. 
Consider the distribution $p_{X_1X_2Y}$ with $(13``13")=(14``14")=(15``15")=(16``16")=\tfrac{1}{8}$, $(23``23")=\tfrac{1}{8}-\delta$, $(24``24")=\tfrac{1}{8}+\delta$, $(25``25")=\tfrac{1}{8}+{\delta }'$, $(26``26")=\tfrac{1}{8}-{\delta }'$, where $\delta ,{\delta }'<\tfrac{1}{8}$ and $(ab``ab")\defeq p_{{X_1}{X_2}{Y}}(a,b,ab)$. 
If $\delta \ne {\delta }'$, as $\delta,{\delta }'\to 0$, we have the following ideal PI decomposition: ${I_ \cap }\left( {\{ {X_1},{X_2}\};({X_{{1}}},{X_{{2}}})} \right) = I({X_{{1}}};{X_{{2}}})\to 0$, $SI(\{ {X_1}{X_2}\};Y)=0$, $UI(\{ {X_1}\};Y)=H(X_1|X_2)\approx +1.0$ and $UI(\{ {X_2}\};Y)=H(X_2|X_1)\approx +2.0$.
\end{example}
\noindent
Consider again the source network with coded side information setup \cite{ref42} where predictors $X_1$ and $X_2$ are independently encoded and a joint decoder wishes to losslessly reconstruct only $X_1$, using the coded $X_2$ as side information. 
It is tempting to assume that a complete description of $X_1$ is always possible by coding the side information at a rate $R_{X_2}=I(X_1;X_2)$ and describing the remaining uncertainty about $X_1$ at rate $R_{X_1}=H(X_1|X_2)$. 
Example 11 provides an interesting counterexample to this intuition. Since the conditional distributions $p_{X_1|X_2}(\cdot|x_2)$ are different for all $x_2 \in {\mathcal{X}}_2$, we have $R_{X_2} \geq H(X_2)$ (see Theorem 2). Consequently one needs to fully describe $X_2$ to (losslessly) recover $X_1$, even if it is the case that $I(X_1;X_2)$ is arbitrarily small. Therefore, separating the redundant and unique information contributions from $X_2$ is not possible in this case.

\subsection{Conclusions}

We first took a closer look at the varied relationships between two RVs. Assuming information is embodied in $\sigma$-algebras and sample space partitions, we formalized the notions of common and private information structures. We explored the subtleties involved in decomposing $H(XY)$ into common and private parts. The richness of the information decomposition problem is already manifest in this simple case in which common and private informational parts sometimes cannot be isolated. 
We also inquired if a nonnegative PI decomposition of the total mutual information can be achieved using a measure of redundancy based on common information. We answered this question in the negative. In particular, we showed that for independent predictor RVs when any nonvanishing redundancy can be attributed solely to a mechanistic dependence between the target and the predictors, any common information based measure of redundancy cannot induce a nonnegative PI decomposition. 

Existing measures of synergistic \cite{ref8} and unique \cite{ref5} information use optimization over three auxiliary RVs to achieve a nonnegative decomposition. We leave as an open question if optimization over two auxiliary RVs can achieve a similar feat.
Also, at present it is not clear if the coding-theoretic interpretation leading up to the counterexample in Example 11 calls into question the bivariate PI decomposition framework itself. More work is needed to assess its implications on the definitions of redundant and unique information.

In closing, we mention two other candidate decompositions of the total mutual information. Pola \emph{et al.} proposed a decomposition of the total mutual information between the target and the predictors into terms that account for different coding modalities \cite{ref55}. Some of the terms can, however, exceed the total mutual information \cite{ref28}. Consequently, the decomposition is not nonnegative, thus severely limiting the operational interpretation of the different coding components.
More recently, a decomposition of the total mutual information is proposed in \cite{ref23} based on a notion of synergistic information, $S^{(2)}$, using maximum entropy projections on $k$-th order interaction spaces \cite{ref23a,ref23b}. The ensuing decomposition is, however, incompatible with \textbf{(LP)} \cite{ref23}. Like $I_\cap^1$, $S^{(2)}$ is symmetric with respect to permutations of the target and the predictor RVs which strongly hints that $S^{(2)}$ fails to capture any notion of mechanistic dependence. Indeed, for the \textsc{And} example, $S^{(2)}$ computes to zero, and consequently \textbf{(LP)} is violated.   

In general, the quest for an operationally justified nonnegative decomposition of multivariate information remains an open problem. Finally, given the subtle nature of the decomposition problem, intuition is not the best guide.

\vspace{-1mm}
\section*{Acknowledgment}
\vspace{1mm}
Thanks are due to Aditya Mahajan for short useful discussions over email.

\section{Appendices}
\subsection{Appendix A: Supplemental proofs omitted in Section II} 
\vspace{2mm}
\begin{proof}[Proof of Lemma 2]
(See Problem 16.25, p. 392 in \cite{ref34}; also see Corollary 1 in \cite{ref33}).
Given $p_{Q|XY}$ such that $X-Y-Q$ and $Y-X-Q$, it follows that $p_{XY}(x,y)>0 \implies p_{Q|XY}(q|x,y)=p_{Q|X}(q|x)=p_{Q|Y}(q|y)$ $\forall q$. Given an ergodic decomposition of $p_{XY}(x,y)$ such that $\mathcal{X}\times\mathcal{Y}=\bigcup\nolimits_{{q'}} \mathcal{X}_{q'}\times\mathcal{Y}_{q'}$, where the $\mathcal{X}_{q'}$$'$s and $\mathcal{Y}_{q'}$$'$s having different subscripts are disjoint, define $p_{Q'|XY}$ as $Q'=q'$ $\iff$ $x \in \mathcal{X}_{q'} \iff y \in \mathcal{Y}_{q'}$. Clearly $H(Q'|X)=H(Q'|Y)=0$. Then, for any $Q=q$ and for every $q'$, $p_{Q|XY}(q|\cdot,\cdot)$ is constant over $\mathcal{X}_{q'}\times\mathcal{Y}_{q'}$ which implies that $p_{Q|XY}(q|x,y)=p_{Q|Q'}(q|q')$. Thus, for any ${q}'$ for which ${{p}_{{{Q}'}}}({q}')>0$, ${p_{XYQ|Q'}}(x,y,q|q') = {p_{Q|XYQ'}}(q|x,y,q'){p_{XY|Q'}}(x,y|q')= {p_{Q|XY}}(q|x,y){p_{XY|Q'}}(x,y|q') = {p_{Q|Q'}}(q|q'){p_{XY|Q'}}(x,y|q')$, so that $XY-Q'-Q$. The converse is obvious. Thus, given (2), we get $Q'$ such that $I(XY;Q|Q')=0$ so that $I(XY;Q)=I(XYQ';Q)=I(Q';Q)=H(Q')-H(Q'|Q)\leq H(Q')$.
\end{proof}

\vspace{3mm}
\emph{Lemma A1.} \emph{${{p}_{XY}}$ is saturable iff there exists a pmf ${{p}_{Q|XY}}$ such that $X-Q-Y,\text{  }Q-X-Y,\text{  }Q-Y-X$.}
\vspace{1mm}
\begin{proof}
Given $Q:X-Y-Q,\text{  }Y-X-Q$, by Lemma 2, there exists a pmf ${{p}_{{Q}'|XY}}$ such that $H(Q'|X)=H(Q'|Y)=0$ and $XY-{Q}'-Q$. Clearly, $I(X;Y|Q)=0\implies I(X;Y|{Q}')=0$ since $I(X;Y|Q)=I(X{Q}';Y|Q)\geq I(X;Y|Q{Q}')=I(X;Y|{Q}')$, where the last equality follows from $XY-{Q}'-Q$. Taking ${Q}'$ as ${{Q}_{*}}$, the claim follows. 
Taking ${{Q}_{*}}$ as $Q$, the other direction is obvious. 
\end{proof}

\vspace{3mm}
\emph{Lemma A2.} \emph{${{C}_{GK}}(X;Y)=I(X;Y)\iff I(X;Y)={{C}_{W}}(X;Y)$} (see Problem 16.30, p. 395 in \cite{ref34}).
\vspace{1mm}
\begin{proof}
Let RV ${{Q}_{1}}$ achieve the minimization in (4). Note the following chain of equivalences \cite{ref33}: $I(XY;{Q_1}) = I(X;Y) \iff H(XY|{Q_1}) = H(X|Y) + H(Y|X) \mathop=\limits^{{\text{(a)}}} H(X|{Q_1}Y)$ $+ H(Y|{Q_1}X)\iff {Q_1}-X-Y,\text{ } {Q_1}-Y-X$, where (a) follows from $X-{{Q}_{1}}-Y.$ The claim follows then from invoking Lemma 2 and noting that $I(XY;{{Q}_{1}})=H({{Q}_{*}})={{C}_{GK}}(X;Y)$, where ${Q}_{*}$ is the maximal common RV.
\end{proof}

\vspace{3mm}
\emph{Lemma A3.} \emph{$C_{GK}(X_1;\ldots;X_K)$ is monotonically nonincreasing in $K$, whereas $C_{W}(X_1;\ldots;X_K)$ is monotonically nondecreasing in $K$. Also $C_{GK}(X_1;\ldots;X_K) \leq \mathop{\min}\limits_{i\ne j} I(X_i;X_j)$, while $C_W(X_1;\ldots;X_K) \geq \mathop{\max}\limits_{i\ne j} I(X_i;X_j)$, for any $i,j\in \{1,\ldots,K\}$.}
\vspace{1mm}
\begin{proof}
Let ${X_\mathcal{A}} \triangleq {\{ {X_i}\} _{i \in \mathcal{A}}}$ be a $K$-tuple of RVs ranging over finite sets $\mathcal{X}_i$ where $\mathcal{A}$ is an index set of size $K$, and let ${{\mathcal{P}}_{{X_\mathcal{A}}}}$ be the set of all conditional pmfs ${p_{Q|{X_\mathcal{A}}}}$ s.t. $|{\mathcal{Q}}| \leq \prod\nolimits_{i = 1}^K {|{\mathcal{X}_i}}| + 2$. First note the following easy extensions.
\begin{align*}
{C_{GK}}({X_1}; \ldots ;{X_K})&=\mathop {\max }\limits_{\substack{Q:\text{ }Q - {X_i} - {X_{\mathcal{A}\setminus i}}, \forall i \in {\mathcal{A}}}} I(X_{\mathcal{A}};Q)\\
{C_W}({X_1}; \ldots ;{X_K})&=\mathop {\min }\limits_{\substack{Q:\text{ }{X_i} - Q - {X_j}, \forall i,j \in {\mathcal{A}},i \ne j}} I(X_{\mathcal{A}};Q)
\end{align*}
Given ${p_{Q|{X_\mathcal{A}}}} \in {\mathcal{P}}_{{X_\mathcal{A}}}$ such that (a) $Q - {X_i} - {X_{\mathcal{A}\setminus i}}, \forall i \in {\mathcal{A}}$, we have $I({X_{\mathcal{A}\setminus K}};Q) \mathop{=}\limits^{{\text{(b)}}} I({X_{\mathcal{A}\setminus K}};Q)+I({X_{\mathcal{A}\setminus 1}};Q|X_1) \mathop{\geq}\limits^{{\text{(c)}}}I({X_\mathcal{A}};Q)$, where (b) follows from using $i=1$ in (a), and (c) follows from noting that $I(X_{\mathcal{A}\setminus 1};Q|X_1)\geq I(X_K;Q|X_{\mathcal{A}\setminus K})$. We then have
\begin{align*}
\mathop {\max }\limits_{\substack{Q:\text{ }Q - {X_i} - {X_{\mathcal{A}\setminus i}},\\\hspace{3mm} \forall i \in {\mathcal{A}}}} I(X_{\mathcal{A}};Q) &\leq \mathop {\max }\limits_{\substack{Q:\text{ }Q - {X_i} - {X_{\mathcal{A}\setminus i}}, \\\hspace{2mm} \forall i \in {\mathcal{A}}}} I(X_{\mathcal{A}\setminus K};Q) \\
&\mathop{\leq}\limits^{{\text{(d)}}} \mathop {\max }\limits_{\substack{Q:\text{ }Q - {X_i} - {X_{\mathcal{A}\setminus{\{i,K\}}}}, \\\hspace{3mm} \forall i \in {\mathcal{A}\setminus K }}} I(X_{\mathcal{A}\setminus K};Q),
\end{align*}
where (d) follows since $\forall i \in {\mathcal{A}}$, $Q - {X_i} - {X_{\mathcal{A}\setminus i}}$ implies $Q - {X_i} - {X_{\mathcal{A}\setminus{\{i,K\}}}}, \forall i \in {\mathcal{A}\setminus K}$. Hence, ${C_{GK}}({X_1}; \ldots ;{X_K})\leq {C_{GK}}({X_1}; \ldots ;X_{K-1})$. Also note that for any $i,j\in \mathcal{A}$,  $I(X_{\mathcal{A}};Q)\mathop{=}\limits^{{\text{(e)}}}I(X_i;Q)\mathop{\leq}\limits^{{\text{(f)}}}I(X_i;X_j)$, where (e) follows from (a) and (f) follows from invoking the data processing inequality after using (a) again, since for any $j \in \mathcal{A}$, $Q - {X_j} - {X_{\mathcal{A}\setminus j}} \implies Q-X_j-X_i$, with $i \in \mathcal{A}\setminus j$. Hence $C_{GK}(X_1;\ldots;X_K) \leq \mathop{\min}\limits_{i\ne j} I(X_i;X_j)$.

The claim for monotonicity of the Wyner CI is immediate from noting that $\mathop {\min }\limits_{\substack{Q:\text{ }{X_i} - Q - {X_j}, \forall i,j \in {\mathcal{A}\setminus K},i \ne j}} I(X_{\mathcal{A}\setminus K};Q) \leq \mathop {\min }\limits_{\substack{Q:\text{ }{X_i} - Q - {X_j}, \forall i,j \in {\mathcal{A}},i \ne j}} I(X_{\mathcal{A}};Q)$, since the constraint set for $Q$ in the $\operatorname{RHS}$ is a subset of that for the $\operatorname{LHS}$. Further, for any $i,j\in \mathcal{A}$, $X_i-Q-X_j \implies I(X_i;X_j)\leq I(X_j;Q)\leq I(X_{\mathcal{A}};Q)$, whence $\mathop{\max}\limits_{i\ne j} I(X_i;X_j) \leq C_W(X_1;\ldots;X_K)$ follows.
\end{proof}

\vspace{3mm}
\begin{proof}[Proof of Lemma 3]
Let $Q_{Y}^{X}={{Q}_{1}}$. 
Clearly, ${{Q}_{1}}-Y-X$ is a Markov chain. 
$Y-{{Q}_{1}}-X$ is also a Markov chain since given ${{Q}_{1}}={{q}_{1}}$, ${p_{X|{Q_1}}}(x|{q_1}) = \sum\nolimits_{y \in {\mathcal{Y}}} {{p_{XY|{Q_1}}}(xy|{q_1})}=\sum\nolimits_{y:{\text{}}{Q_1}={q_1}}{{p_{Y|{Q_1}}}(y|{q_1})}{p_{X|Y{Q_1}}}(x|y{q_1})$ $={p_{X|Y{Q_1}}}(x|y{q_1})$, $\forall y{\text{ given }}{Q_1} = {q_1}$.
Now let ${{Q}_{2}}=g(Y)$ so that $X-{{Q}_{2}}-Y$. For some $y,{y}'\in {\mathcal{Y}}$, let $g(y)=g({y}')={{q}_{2}}$. Then, ${{p}_{X|{{Q}_{2}}}}(x|{{q}_{2}})$ $={{p}_{X|Y}}(x|y)={{p}_{X|Y}}(x|{y}'),\text{ }x\in {\mathcal{X}}$. Thus $f(y)=f({y}')$ which implies $X-{{Q}_{1}}-{{Q}_{2}}-Y$. Hence ${{Q}_{1}}=Q_{Y}^{X}$ is a minimal sufficient statistic of $Y$with respect to $X$.
\end{proof}

\vspace{3mm}
\begin{proof}[Proof of Theorem 1]
From (7) and Lemma 3, it follows that $Q_{Y}^{X}$ is the minimizer in ${{\tilde{C}}_{W}}(Y\backslash X)$. Since $Q_{Y}^{X}$ is a minimal sufficient statistic of $Y$ with respect to $X$, for any $Q$ s.t. $Q-X-Y$ and $X-Y-Q$, it follows that $H(Q_{Y}^{X}|Q)=0$ (also see Lemma 3.4(5) in \cite{ref43}). Thus, $Q_{Y}^{X}$ achieves the maximum in ${{\tilde{P}}_{W}}(Y\backslash X)$. The decomposition $H(Y)=H(Y|Q_{Y}^{X})+H(Q_{Y}^{X})$ easily follows. Finally, $\mathop {\min }\limits_{Q:X - Q - Y} I(XY;Q)$ $\leq$ $\mathop {\min }\limits_{\substack{ Q:\text{ }X-Q-Y,\hspace{1mm}X-Y-Q}} H(Q)$, 
because if $X-Y-Q$ then $I(XY;Q)=I(Q;Y)\leq H(Q)$, so that ${{C}_{W}}(X;Y)\leq {{\tilde{C}}_{W}}(Y\backslash X)$.                       
\end{proof}

\vspace{3mm}
\begin{proof}[Proof of Theorem 2]
When $p_{XY}$ lacks the structure to form components of size greater than one, it follows from Lemma 3 that ${{\tilde{C}}_{W}}(Y\backslash X)=H(Y)$. Consequently by Theorem 1, ${{\tilde{P}}_{W}}(Y\backslash X)=0$. For the other direction, see Corollary 3 in \cite{ref42}, where analogous bounds for ${{\tilde{C}}_{W}}(Y\backslash X)$ are given.

For the second part, we prove the first equivalence. Let ${{Q}_{*}}=X\wedge Y={{g}_{X}}(X)={{g}_{Y}}(Y)$. For $x\in \mathsf{\mathcal{X}},\text{ }$ if $y,{y}'\in \mathsf{\mathcal{Y}}$ do not induce different conditional distributions on $X$, i.e., if ${{p}_{X|Y}}(x|y)={{p}_{X|Y}}(x|{y}')$, then we must have ${{g}_{Y}}(y)={{g}_{Y}}({y}')$. This implies the existence of a function $f$ such that ${{Q}_{*}}=f(Q_{Y}^{X})$. If ${{p}_{XY}}$ is saturable, we also have $X-{{Q}_{*}}-Y$. From Lemma 3 and Remark 3, it then follows that $H({{Q}_{*}})=H(Q_{Y}^{X})=I(X;Y)$ and consequently ${{\tilde{P}}_{W}}(Y\backslash X)=H(Y|X)$. For the converse, first note that $H({{Q}_{*}})\le I(X;Y)\le {{C}_{W}}(X;Y)\le H(Q_{Y}^{X})$ (see Remark 3 and (8d)). Demanding ${{\tilde{P}}_{W}}(Y\backslash X)=H(Y|X)$ or equivalently $H(Q_{Y}^{X})=I(X;Y)$ implies $I(X;Y)={{C}_{W}}(X;Y)=H(Q_{Y}^{X})$ when from Remark 3 it follows that ${{p}_{XY}}$ is saturable. For the second equivalence, see Theorem 4 in \cite{ref42}. This concludes the proof.
\end{proof}

\vspace{3mm}
\subsection{Appendix B}
We briefly provide several examples and applications, where information-theoretic notions of synergy and redundancy are deemed useful.

\emph{Synergistic and redundant information}. 
Synergistic interactions in the brain are observed at different levels of description. At the level of brain regions, cross-modal illusions offer a powerful window into how the brain integrates information streams emanating from multiple sensory modalities \cite{ref12}. A classic example of synergistic interaction between the visual and auditory channels is the Mcgurk illusion \cite{ref13}. Conflicting voice and lip-movement cues can produce a percept that differs in both magnitude and quality from the sum of the two converging stimuli. At the single neuronal level, temporally and spatially coincident multimodal cues can increase the firing rate of individual multisensory neurons of the superior colliculus beyond that can be predicted by summing the unimodal responses \cite{ref14}. In the context of neural coding, a pair of spikes closely spaced in time can jointly convey more than twice the information carried by a single spike \cite{ref15}. 
In cortex studies, evidence of weak synergy have been been found in the somatosensory \cite{ref55} and motor \cite{ref90} and primary visual cortex \cite{ref56}.
Similarly, there are several studies evidencing net redundancy at the neuronal population level \cite{ref28,ref29,ref51,ref54,ref55,ref56,ref53,ref58}. 
Often studies on the same model system have reached somewhat disparate conclusions. 
For instance, retinal population codes have been found to be approximately independent \cite{ref58b}, synergistic \cite{ref58a}, or redundant \cite{ref58}.

\emph{Unique information}. 
A wealth of evidence suggests that attributes such as color, motion and depth are encoded uniquely in perceptually separable channels in the primate visual system \cite{ref16}, \cite{ref12}. The failure to perceive apparent motion with isoluminant colored stimuli, dubbed as the color-motion illusion \cite{ref16} demonstrates that the color and motion pathways provide unique information with respect to each other. There is also mounting evidence in favor of two separate visual subsystems \cite{ref17} that encode the allocentric (vision for perception) and egocentric (vision for action) coordinates uniquely along the ventral and the dorsal pathways, respectively, for object identification and sensorimotor transformations.

In embodied approaches to cognition, an agent's physical interactions with the environment generates structured information and redundancies across multiple sensory modalities that facilitates cross-modal associations, learning and exploratory behavior \cite{ref18}. More recent work has focused on information decomposition in the sensorimotor loop to quantify morphological computation which is the contribution of an agent's morphology and environment to its behavior \cite{ref20}. Some related decompositions have also focused on extracting system-environment boundaries supporting biological autonomy \cite{ref21}.

Further motivating examples for studying information decomposition in general abound in cryptography \cite{ref25}, distributed control \cite{ref26} and adversarial settings like game theory \cite{ref27}, where notions of common knowledge shared between agents are used to describe epistemic states.

\end{document}